\documentclass[12pt,draftcls,onecolumn]{IEEEtran}

\usepackage{amsbsy,amsmath,amsfonts,amssymb,amsbsy,subfigure,booktabs,multirow}
\usepackage{bm,cite,graphicx,psfrag,pstricks,theorem,tikz,times,url,verbatim}
\usetikzlibrary{shapes,snakes,calendar,matrix,backgrounds,folding}
\usepackage[english]{babel}
\usetikzlibrary{arrows,automata}
\usetikzlibrary{calc} % for manimulation of coordinates
\usetikzlibrary{decorations.pathmorphing} % for snake lines

\newcommand{\bi}{\begin{itemize}}
\newcommand{\ei}{\end{itemize}}
\newcommand{\ben}{\begin{enumerate}}
\newcommand{\een}{\end{enumerate}}

\newcommand{\bc}{\begin{cases}}
\newcommand{\ec}{\end{cases}}
\newcommand{\bd}{\begin{description}}
\newcommand{\ed}{\end{description}}

\newcommand{\be}{\begin{equation}}
\newcommand{\ee}{\end{equation}}
\newcommand{\bea}{\begin{eqnarray}}
\newcommand{\eea}{\end{eqnarray}}

\newtheorem{thm}{Theorem}

\newtheorem{propos}{Proposition}

\newtheorem{ass}{Assumption}

\newtheorem{algo}{Algorithm}

\theoremstyle{plain}

\newcommand{\nn}{\nonumber}

\graphicspath{%
     {./Figures/}
}

\begin{document}

\title{Capacity of electron-based communication over bacterial cables: the full-CSI case}
%\vspace{-9mm}
\author{Nicol\`{o}~Michelusi~and~Urbashi~Mitra
%\vspace{-10mm}
\thanks{%\vspace{-5mm}\newline 
This research has been funded in part by the following grants:
ONR N00014-09-1-0700, AFOSR FA9550-12-1-0215, DOT CA-26-7084-00, NSF CCF-1117896, NSF CNS-1213128, NSF CCF-1410009, NSF CPS-1446901.
N. Michelusi is in part supported by AEIT (Italian association of electrical engineering) through the scholarship "Isabella Sassi Bonadonna 2013".
}
\thanks{N. Michelusi and U. Mitra are with the Dept. of Electrical Engineering, University of Southern California. email: \{michelus,ubli\}@usc.edu.}
\thanks{Part of this work will appear at the 2015 IEEE International Conference on Communications \cite{MichelusiICC15}
and has been submitted to the 2015 IEEE International Symposium on Information Theory (ISIT) \cite{MichelusiISIT15}.
}
}
\maketitle
\vspace{-5mm}
\begin{abstract}
Motivated by recent discoveries of microbial communities that transfer electrons across centimeter-length scales, this paper studies the information capacity of bacterial cables via electron transfer, which coexists with molecular communications,
under the assumption of full causal channel state information (CSI).
 The bacterial cable is modeled as an electron queue that transfers electrons from the encoder at the electron donor source, which controls the desired input electron intensity, to the decoder at the electron acceptor sink. \emph{Clogging} due to local ATP saturation along the cable is modeled. A discrete-time scheme is investigated, enabling the computation of an achievable rate. The regime of asymptotically small time-slot duration is analyzed, and the optimality of binary input distributions is proved, \emph{i.e.}, the encoder transmits at either maximum or minimum intensity, as dictated by the physical constraints of the cable. A dynamic programming formulation of the capacity is proposed, and the optimal binary signaling is determined via policy iteration. It is proved that the optimal signaling has smaller intensity than that given by the \emph{myopic policy}, which greedily maximizes the instantaneous information rate but neglects its effect on the steady-state cable distribution. In contrast,
 the optimal scheme balances the tension between achieving high instantaneous information rate, and inducing a favorable steady-state distribution, such that those states characterized by high information rates are visited more frequently, thus revealing the importance of CSI. This work represents a first contribution towards the design of electron signaling schemes in complex microbial structures,
  \emph{e.g.}, bacterial cables and biofilms, 
 where the tension between maximizing the transfer of information and guaranteeing the well-being of the overall bacterial community arises,
 and motivates further research on the design of more practical schemes, where CSI is only partially available.
\end{abstract}
\section{Introduction}
Cellular respiration relies on a continuous flow of electrons from an electron donor (ED) to an electron acceptor (EA) through the electron transport chain (ETC) of the cell to produce energy in the form of the molecule adenosine triphosphate (ATP), and to sustain vital operations and functions \cite[Section 16.2]{Lodish}.
This fundamental mechanism is well-known for individual, isolated cells. However, in the past decade,
 remarkable discoveries of multi-cellular microbial communities that transfer electrons between cells and across 
multi-cellular structures have been made \cite{Naggar}.
This mechanism, termed \emph{electron transfer}, has been observed
in molecular assemblies known as \emph{bacterial nanowires}, and in macroscopic architectures, such as biofilms and multi-cellular bacterial cables \cite{Pfeffer}.
  These experimental observations raise the possibility of 
  microbial communication via \emph{electron transfer}, which coexists with the better-known communication strategies based on molecular diffusion \cite{Naggar,Reguera2,Arjmandi},
  enabling cells to quickly sense and respond to their environment.

In this paper, motivated by these experimental observations and building on our
recent queuing theoretic model of bacterial cables \cite{JSACmiche},
we study the capacity of bacterial cables under the assumption of perfect CSI at both the encoder and the decoder.
 The encoder controls the desired intensity of the electron signal entering the cable,  and the decoder attempts to decode  the transmitted message based on the measured output electron process.
  The bacterial cable (the channel) has an internal state evolving dynamically as a function of electrons leaving and entering the cable,
 and events occurring within the cable, and therefore falls within the broad class of channels with state, \emph{e.g.}, \cite{Goldsmith,Pfister,Chen}.
 Specifically,
the cable is treated as a single black box, which takes electrons as input from the ED source and outputs electrons into the EA sink,
and is modeled as a finite-state Markov channel~\cite{Chen}, controlled by the input signal.
That is, the detailed dynamics of electrons and of the states of the cells located along the cable, as discussed in \cite{JSACmiche},
are not explicitly accounted for, but only \emph{global} effects on the cable resulting from these \emph{local} interactions are modeled.
 Electrons enter and exit the cable according to Poisson processes,
 whose intensities are functions of the internal cable state, which in turn determines the ability of the cable to relay electrons,
 and of the encoded signal.
 
The contributions of this paper are as follows.
 We characterize the capacity of a discrete-time version of the system,
 using results on the capacity of finite-state Markov channels with feedback derived in \cite{Chen},
  and prove the optimality of stationary Markov input distributions, which depend only on the current state of the cable.
  We then consider 
the regime of asymptotically small time-slot duration. Based on this asymptotic analysis, we prove the optimality of  
binary input distributions, thus extending previous results on the capacity of \emph{static} Poisson channels \cite{Wyner} to our \emph{dynamic} setting.
We prove that the capacity maximization problem is a Markov decision process (MDP) \cite{Bertsekas2005},
with state given by the internal cable state, action given by the 
expected desired input electron intensity, which generates the binary intensity signal,
 and reward given by the instantaneous mutual information rate,
and can thus be solved efficiently using standard optimization algorithms, \emph{e.g.}, policy iteration (see \cite{Bertsekas2005}).
We show that the optimal input distribution optimizes a trade-off between
achieving high instantaneous information rate, and inducing
an "optimal" steady-state distribution of the cable state, 
such that those states
where the transmission of information is more favorable are visited more frequently.
On the other hand, 
the optimal distribution for the \emph{myopic} or \emph{greedy} policy, which greedily maximizes the instantaneous information rate, without considering 
its impact on the steady-state distribution of the cable, performs poorly.
In particular, we prove that the optimal expected desired input electron intensity
is smaller than that dictated by the myopic strategy. 
In fact, larger input electron intensities 
tend to quickly recharge the electron reserves within the cable, hence the ATP reserves of the cells,
resulting in a clogging of the cable rendering it unable to further relay electrons until
its reserves discharge to sufficiently low levels.

Remarkably, this work represents a first contribution towards the design of \emph{electron signaling} schemes in complex microbial structures, \emph{e.g.}, bacterial cables and biofilms,
 where the tension between maximizing the transfer of information and guaranteeing the well-being of the overall bacterial community  often arises,
 and thus motivates further research in this direction, \emph{e.g}, using methods based on statistical physics~\cite{Mitra}.
 Moreover, our numerical evaluations
 reveal the importance of CSI, which enables adaptation
of the input signal to the state of the cell, and thus motivates further research on the design of more practical schemes, where CSI is only partially available,
and of state estimation techniques.

Most of the recent literature
on the design of biological communication systems is based on 
\emph{molecular diffusion} \cite{Arjmandi,Nakano,Rose,Kadloor,Einolghozati,Mosayebi,Oiwa,Kuran201086}.
The achievable capacity for the chemical channel
is investigated in \cite{Kadloor}, under Brownian motion, and in \cite{Einolghozati}, under a diffusion channel.
In \cite{Arjmandi,Mosayebi}, novel molecular modulations are proposed.
In \cite{Oiwa}, an in-vitro molecular communication system  is designed and,
in \cite{Kuran201086}, an energy model is proposed.
In \cite{JSACUM2}, 
upper bounds on the capacity of communication networks over linear time-invariant  Poisson channels
have been investigated in the context of molecular diffusion, based on the symmetrized Kullback-Leibler divergence.
Therein, a static channel with inter-symbol interference is considered, and, similar to our work, the optimality of binary input distributions is also proved based on such upper bound.
In this paper, instead, we
build on the stochastic model recently developed in \cite{JSACmiche}
to study the capacity of bacterial cables via \emph{electron transfer}.
The proposed channel model is \emph{dynamic}, as opposed to \emph{static}, and
 explicitly captures biological constraints 
of bacterial cables, which are not present in microbial communications based on molecular diffusion,
such as clogging of the cable induced by local ATP saturation of the cells, resulting in 
degradation of the electron transfer efficiency of the cable, and the minimum input electron flow requirement in order to keep the cells alive.

Finite-state Markov channels have received significant attention in the information theoretic community, \emph{e.g.}, \cite{Goldsmith,Pfister,Chen}.
The case with feedback and CSI at both the encoder and decoder is considered in \cite{Chen}. We specialize the capacity formulation of \cite{Chen} to our Poisson channel, and show that 
it can be achieved by input distributions independent of the feedback signal.
The case with no CSI has been considered in \cite{Pfister} and in
 \cite{Goldsmith}, for finite-state Markov channels whose
 transition probabilities are independent of the input signal.
 The capacity of channels with i.i.d. states
 controlled by actions scheduled by the encoder
  has been considered in \cite{Weissman}, and in \cite{Chiru}, for the setting where actions are 
  generated adaptively.

The capacity of continuous-time Poisson channels is known, and it has been derived 
 in \cite{Wyner} for the \emph{static} case.
The case with i.i.d. block-fading, CSI at the receiver, and partial CSI at the transmitter,
has been considered in \cite{Chakraborty}.
The capacity of a Poisson channel with side information 
on spurious counts generated by an adversary at the receiver
is considered in \cite{Bross}. %It is shown that side information does not increase capacity.
The bacterial cable considered in this paper is also
modeled as a continuous-time Poisson channel. However, unlike \cite{Wyner,Chakraborty,Bross},
the channel is finite-state Markov, and its state is controlled by the input signal generated by the encoder.
While the capacity of continuos-time channels is known,
the capacity of discrete-time Poisson channels is unknown,
and only upper and lower bounds have been derived \cite{JSACUM2,Lapidoth}.

This paper is organized as follows. In Section~\ref{sysmo}, we present the system model
and the discrete-time representation.
In Section~\ref{sec:capacity}, we analyze the capacity of the discrete-time model and
study its asymptotic capacity.
In Section~\ref{numres}, we present numerical results.
In Section~\ref{concl}, we conclude the paper. 
The proofs of the theorems and propositions are provided in the Appendix.

\section{System model}
\label{sysmo}
We consider a continuous-time communication system based on a bacterial cable, depicted in Fig. \ref{baccable}.
A stochastic model for a bacterial cable has been proposed in \cite{JSACmiche}.
The bacterial cable
contains $C$ cells and 
 has an internal state $S(t)\in\mathcal S$ at time $t$,
 taking values from the state space $\mathcal S$. $S(t)$ evolves in a stochastic fashion as a result of random
events occurring within the cable, \emph{e.g.}, electrons entering ($X(t)$) and exiting ($Y(t)$) the cable at random times,
as detailed in~\cite{JSACmiche}.
The communication system includes an encoder,
which maps the message $m\in\{1,2,\dots,M\}$ to a \emph{desired electron intensity} $\lambda(t)\in [\lambda_{\min},\lambda_{\max}],\ t\in[0,T]$,
where
$T$ is the codeword duration,
and $\lambda_{\min}>0$ and $\lambda_{\max}>\lambda_{\min}$ are, respectively, the minimum and maximum electron intensities allowed into the cable.
These are physical constraints induced by the nature of the cable.
We let $\rho\triangleq\frac{\lambda_{\min}}{\lambda_{\max}}\leq 1$.
Note that the stochastic model developed in \cite{JSACmiche} assumes that the cells 
may die at random times, \emph{e.g.}, as a consequence of an insufficient input electron flow.
In this paper, we assume that the minimum input electron intensity
$\lambda_{\min}$ is sufficient to keep the entire cable alive, so that  cells do not die.\footnote{The
 more general case where cells die at random times
will be considered as future work, and can be analyzed using tools developed in \cite{Varshney}.}
Note that $\lambda_{\min}$ can also be interpreted as the \emph{dark current} \cite{Wyner}.
Electrons enter the bacterial cable following a Poisson process with rate $\lambda_{in}(t){=}A(S(t))\lambda(t)$,
where $A(s)\in[0,1]$ is the \emph{clogging state}, and denotes the ability of the bacterial cable to accept electrons.
$A(s)$ is specific to our bacterial cable model, and is not present, \emph{e.g.}, in static Poisson channels, where $A(s)=1$ \cite{Wyner}.
In particular, $A(s){=}0$ if the bacterial cable is clogged and no electrons can be accepted, and 
$A(s){=}1$ if all electrons can be accepted without any loss.
In general, $A(s)$ takes value in $[0,1]$, so that only a fraction of the input electrons can be accepted, depending on the state of the cable.
 For instance, if 
the high energy external membrane of the first cell in the cable is full \cite{JSACmiche}, then no electrons can be accepted and 
$A(s)=0$ accordingly.
The input electron process is modeled as a counting process $X(t)\in\{0,1,2,\dots\}$, with $X(0)=0$.
Note that the encoder does not have full control of the timing and number  of electrons released into the cable ($X(t)$), 
but only of $\lambda(t)$,
so that $X(t)$ is a random quantity, which depends both on the encoded signal $\lambda(t)$
 and on the clogging state $A(S(t))$.
 The output electron flow is also modeled as a Poisson process with state-dependent intensity $\mu(S(t))$,
which is a function of the internal state of the cable. We let $Y(t)\in\{0,1,2,\dots\}$ be the counting process associated with the output electron flow,
with $Y(0)=0$, and we denote the maximum output intensity as $\mu_{\max}\triangleq\max_s\mu(s)$.

 \begin{figure}[t]
\centering
\includegraphics[width = .8\linewidth,trim = 10mm 4mm 10mm 9mm,clip=false]{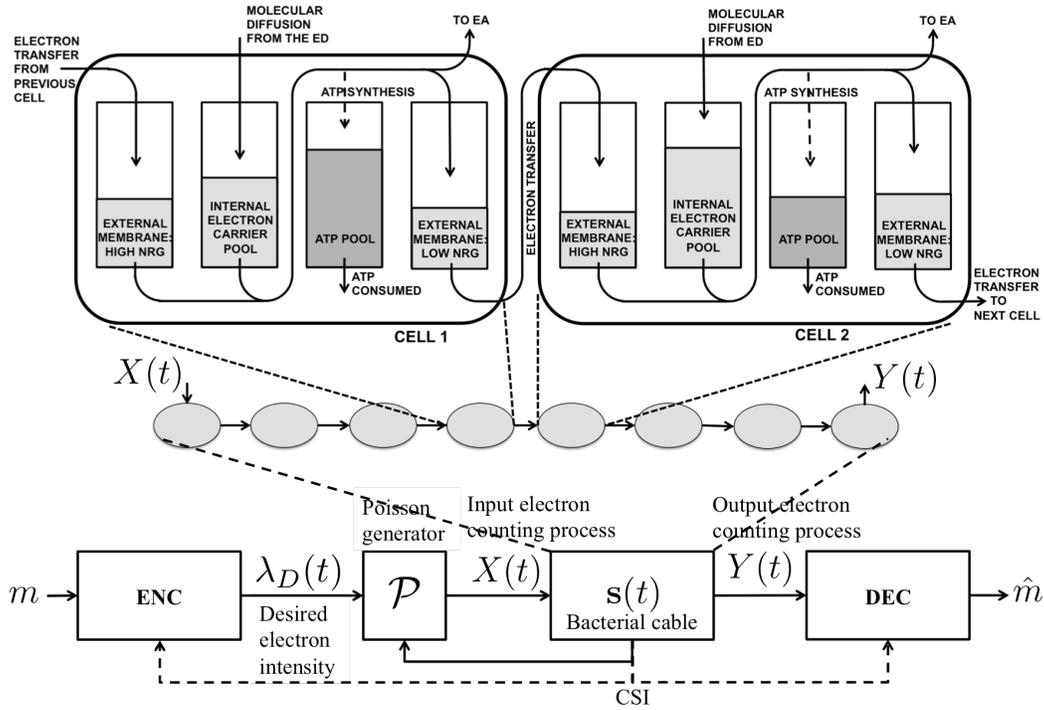}
\caption{Communication system over a bacterial cable.}%\vspace{-4mm}
\label{baccable}
\end{figure}

We assume that both the encoder and the decoder have full CSI,
\emph{i.e.}, the state sequence $S(0:t)=\{S(\tau),\tau<t\}$ is known at both the encoder and the decoder at time $t$.
Therefore, the desired electron intensity at time $t$, $\lambda(t)$, is chosen as a function of the
message $m$ and the CSI  $S(0{:}t)$, $\lambda(t)=f_t(m;S(0{:}t)){\in}[0,\lambda_{\max}]$.
On the other hand, the decoder, given the output 
electron counting process $Y(0{:}T){=}\{Y(t),t{\in}[0,T]\}$
and the cable state sequence $S(0{:}T)$, estimates
 the message $m$
as $\hat m{=}g(Y(0{:}T),S(0{:}T))$.

\subsection{Bacterial cable model}
In this section, we describe a model for the bacterial cable state $S(t)$.
The model \cite{JSACmiche} presumes that $S(t)$ is given by the interconnection of the
states of each cell in the cable, leading to high dimensionality.
In fact, letting $\mathcal S_{cell}$
be the state space of the internal state of each cell, then the overall state space  of the cable is
$\mathcal S\equiv\mathcal S_{cell}^C$, which grows exponentially with the bacterial cable length $C$.
Herein, we propose an approximation of the model presented in \cite{JSACmiche}, 
which treats the bacterial cable as a black box, and captures only the \emph{global} effects on the 
electron transfer efficiency of the cable, resulting from \emph{local} interactions of the cells along the cable.
Specifically,
we model  $S(t)$ as the number of electrons carried in the cable, \emph{i.e.}, the sum of the number of electrons carried
in the external membrane of each cell, which participate in the ETC to produce
ATP for the cell.
The state space for this approximate model is $\mathcal S\equiv\{0,1,\dots, S_{\max}\}$,
where $S_{\max}$ is the electron carrying capacity of the bacterial cable. 
Letting $ S_{\max}^{(cell)}$ be the electron carrying capacity of a single cell,
we have that  $S_{\max}=C\cdot S_{\max}^{(cell)}$.
The rationale behind this approximate model is as follows:
when $S(t)$ is large, the large number of electrons in the cable can sustain 
a large ATP production rate, so that the ATP pools of the cells are full, and the cable is clogged, resulting in $A(S(t))\simeq 0$.
On the other hand, when $S(t)$ is small, a weak electron flow occurs along the cable, so that
the ATP pools are almost empty and the cells are energy-deprived, so that
the cable can sustain a large input electron flow to recharge the ATP pools ($A(S(t))\simeq 1$).
In order to capture this behavior, we make the following assumptions.
\begin{ass}
\label{ass1}
The clogging state $A(s)$ is a non-increasing function of $s$,
with $A(S_{\max}){=}0$, $A(0){=}1$, $A(s){>}0,\ \forall s{<}S_{\max}$.
The output electron intensity, $\mu(s)$, is 
such that $\mu(0){=}0$, $\mu(s){>}0,\forall s{>}0$.
\end{ass}
According to this simplified model, from state $S(t)$ at time $t$, the state of the cable becomes $S(t^+)=S(t)+1$ with rate $\lambda_{in}(t)$,
corresponding to the arrival of one electron into the cable, 
and to state $S(t^+)=S(t)-1$ with rate $\mu(S(t))$, corresponding to one electron exiting the cable.

\subsection{Discrete-time model}
The bacterial cable can be modeled as a finite-state continuous-time Markov channel \cite{Chen} with Poisson input and output.
We now design a discrete-time system, 
enabling the computation of an achievable rate for the continuous-time system.
In particular, we use an approach similar to \cite{Wyner}, which
focused on the capacity analysis
of a static Poisson channel, and we extend it to our finite-state Markov channel.
We divide the codeword duration $T$ into $N$ slots of fixed duration $\Delta=T/N$.
The $k$th slot is the interval $[k\Delta,(k+1)\Delta)$, for $k\in\{0,1,\dots,N-1\}$.
In \cite{Wyner}, it is shown that capacity-achieving schemes for this discrete-time model
asymptotically achieve the capacity of the continuous-time Poisson channel, 
for asymptotically small values of the slot-duration $\Delta\to 0$.
Similarly, we first design capacity-achieving  schemes and the corresponding capacity for the discrete-time model, 
and then analyze the asymptotic regime $\Delta\to 0$ in Section~\ref{sec:capacity}.
The rationale behind this analysis is that,
if
\begin{align}
\label{deltacond}
\Delta\ll 
\frac{1}{
\mu_{\max}+\lambda_{\max}
},
\end{align}
then the probability that more than one event (\emph{i.e.}, multiple electrons entering/leaving the cable) occurs in a single slot is very small, of the order of $o(\Delta^2)$,
whereas the following events are most likely to occur:
1) no electrons enter/leave the cable, with probability 
$\simeq1{-}\Delta \mu(s){-}\Delta A(s)\lambda$;
2) one electron enters the cable, with probability $\simeq \Delta A(s)\lambda$;
3) one electron exits the cable, with probability $\simeq \Delta\mu(s)$.
However, the following analysis holds for any $\Delta>0$.
Let 
\begin{align}
&\alpha_k\triangleq X((k+1)\Delta)-X(k\Delta),
\\
&\beta_k\triangleq Y((k+1)\Delta)-Y(k\Delta),
\end{align}
be the number of electrons entering and exiting the cable in slot $k$, respectively.
We make the following assumptions, similar to \cite{Wyner}:
\begin{enumerate}
\item 
$\lambda(t)$ is constant within each slot.
We denote its constant value in slot $k$ as 
 $\lambda_{k}$, so that  
$\lambda(t){=}\lambda_{k},\forall t\in[k\Delta,(k{+}1)\Delta)$;
\item The encoder and decoder, at the beginning of slot $k$, know the sampled CSI time-series
$S_0^k{=}(S_0,S_1,\dots,S_k)$, where
 $S_k{=}S(k\Delta)$,
rather than the continuous time-series $S(0{:}k\Delta)$;
\item Additionally, the receiver observes $\beta_k$.
However, it assumes $\beta_k>1$ 
 is a rare event. This event is indeed rare when $\Delta$ satisfies condition (\ref{deltacond}),
 since its probability is of the order of $o(\Delta^2)$.
Therefore, the receiver utilizes only the \emph{positive presence of electrons}
 $\hat\beta_k=\chi(\beta_k>0)$
for decoding purposes,
where $\chi(\cdot)$ is the indicator function,
 \emph{i.e.},
$\hat\beta_k=0$ if no electrons are received in slot $k$, and $\hat\beta_k=1$ otherwise.
\end{enumerate}
We let, for $a,b,\geq 0$,
\begin{align}
\!\!\!\!p_{A,B}^{(\Delta)}(a,b|s,\lambda){=}\mathbb P(\alpha_k=a,\beta_k=b|S_k=s,\lambda_{k}=\lambda).
\end{align}
From the properties of Poisson processes,
we have that
\begin{align}
\label{Pab}
&p_{A,B}^{(\Delta)}(a,b|s,\lambda)
=
\left\{
\begin{array}{ll}
1{-}\Delta \mu(s){-}\Delta A(s)\lambda{+}g_{0,0}(\Delta,s,\lambda) & a=0,b=0,\\
\Delta A(s)\lambda+g_{1,0}(\Delta,s,\lambda) & a=1,b=0,\\
\Delta \mu(s)+g_{0,1}(\Delta,s,\lambda)& a=0,b=1,\\
g_{a,b}(\Delta,s,\lambda) & a>0,b>0,
\end{array}
\right.
\end{align}
where $g_{a,b}(\Delta,s,\lambda)$
are functions such that $g_{a,b}(\Delta,s,\lambda)\sim o(\Delta^{\max\{a+b,2\}})$,
\emph{i.e.}, they decay to zero as $\Delta^{\max\{a+b,2\}}$ when $\Delta\to 0$,
and $\sum_{a,b}g_{a,b}(\Delta,s,\lambda)=0$, so that $p_{A,B}^{(\Delta)}(a,b|s,\lambda)$ is a probability distribution. 
%Moreover, for $a>0,b>0$ we have that 
%\begin{align}
%0<\gamma_{a,b}(s,\lambda)<\frac{\lambda_{\max}^a}{a!}\frac{\mu(S_{\max})^b}{b!}.
%\end{align}
The interpretation of  Eq. (\ref{Pab}) is that, when $\Delta$ satisfies the condition in Eq. (\ref{deltacond}),
the event $(a,b)=(0,0)$ (no electrons enter/leave the cable)
occurs with probability $\simeq1{-}\Delta \mu(s){-}\Delta A(s)\lambda$;
the event $(a,b)=(1,0)$ (one electron enters the cable, no electrons leave the cable)
occurs with probability $\simeq \Delta A(s)\lambda$;
whereas
the event $(a,b)=(0,1)$ (no electrons enter the cable, one electron leaves the cable)
occurs with probability $\simeq \Delta \mu(s)$;
all the other events occur with probability of the order of $o(\Delta^2)$.
We let
\begin{align}
p_{S}^{(\Delta)}(s_1|s_0,\lambda)=\mathbb P(S_{k+1}=s_1|S_k=s_0,\lambda_{k}=\lambda).
\end{align}
 Since $S_{k+1}=S_k+\alpha_k-\beta_k$, we have that
\begin{align}
p_{S}^{(\Delta)}(s_1|s_0,\lambda)
=\!\!\!\!\!\!\!\!
\sum_{a=(s_1-s_0)^+}^{\infty}
p_{A,B}^{(\Delta)}(a,s_0-s_1+a|s,\lambda),
\end{align}
and therefore, using (\ref{Pab}),
\begin{align}
\label{ps}
&p_{S}^{(\Delta)}(s_0+1|s_0,\lambda)
{=}
\Delta A(s)\lambda{+}\sum_{a=1}^{\infty}g_{a,a-1}(\Delta,s,\lambda),
\\
&p_{S}^{(\Delta)}(s_0-1|s_0,\lambda)
{=}\Delta \mu(s){+}\sum_{a=0}^{\infty}g_{a,a+1}(\Delta,s,\lambda),
\\
&p_{S}^{(\Delta)}(s_0|s_0,\lambda)
{=}
1{-}\Delta \mu(s){-}\Delta A(s)\lambda
%\nonumber\\&\qquad\qquad\qquad\qquad
{+}\sum_{a=0}^{\infty}g_{a,a}(\Delta,s,\lambda),\!\!
\label{10}
\\
&p_{S}^{(\Delta)}(s_1|s_0,\lambda)
{=}
\!\!\!\!\!\!\!\!\sum_{a=(s_1-s_0)^+}^{\infty}\!\!\!\!\!\!\!\!g_{a,s_0{-}s_1{+}a}(\Delta,s,\lambda),\ 
\forall s_1{\notin}\{s_0{-}1,s_0,s_0{+}1\}.
\label{11}
\end{align}
Finally, we define
the joint probability of state transition $S_k{\to}S_{k+1}$ and channel output $\hat\beta_k$ as
 \begin{align}
\!\!\!\!\! p_{S,\hat B}^{(\Delta)}(s_1,b|s_0,\lambda){=}\mathbb P(S_{k+1}{=}s_1,\hat\beta_k{=}b|S_k{=}s_0,\lambda_{k}{=}\lambda).
 \end{align}
Using the fact that $S_{k+1}=S_k+\alpha_k-\beta_k$, we obtain
\begin{align*}
&p_{S,\hat B}^{(\Delta)}(s_1,0|s_0,\lambda)=p_{A,B}^{(\Delta)}(s_1-s_0,0|s_0,\lambda)\chi(s_1\geq s_0),
\\&
p_{S,\hat B}^{(\Delta)}(s_1,1|s_0,\lambda)=\!\!\!\!\!\!\!\!\!\!\sum_{b=\max\{s_0-s_1,1\}}^{\infty}\!\!\!\!\!\!p_{A,B}^{(\Delta)}(b{+}s_1{-}s_0,b|s_0,\lambda).
\end{align*}
The encoding and decoding schemes for this discrete time model
 are defined as follows:
\begin{itemize}
\item\textbf{Encoding}: in slot $k$, given the message $m$ to be transmitted,
and the CSI $S_0^{k}$,
the encoder defines the desired electron intensity
$\lambda_{k}=f_k(m;S_0^{k})$, taking values from the continuous alphabet $[\lambda_{\min},\lambda_{\max}]$;
\item\textbf{Decoding}: given
the CSI $S_0^{N}$ and the output sequence $\hat\beta_0^{N-1}$,
the decoder returns
the message estimate $\hat m{=}g(S_0^{N},\hat\beta_0^{N-1})$.
\end{itemize}

Note that, similar to \cite{Wyner}, we have converted the original continuous-time Poisson channel into a 
discrete-time channel, with input $\lambda_{k}{\in}[\lambda_{\min},\lambda_{\max}]$ and binary output $\hat\beta_k\in\{0,1\}$.
However, unlike \cite{Wyner} which focuses on a static Poisson channel,
we now have a \emph{finite-state Markov channel}, controlled by the input sequence,
with transition probability $p_{S}^{(\Delta)}(s_1|s_0,\lambda_{k})$ and output distribution
\begin{align*}
\mathbb P(\hat\beta_k=1|S_k=s_0,S_{k+1}=s_1,\lambda_{k})=
\frac{p_{S,\hat B}^{(\Delta)}(s_1,1|s_0,\lambda_{k})}{p_{S}^{(\Delta)}(s_1|s_0,\lambda_{k})}.
\end{align*}
Therefore, the optimal input distribution, which is analyzed in the next section, optimizes a trade-off between
achieving high instantaneous information rate, and inducing transitions in the future to states characterized by large information rate.

\section{Capacity analysis}
\label{sec:capacity}
We have the following proposition, which states an important property of the controlled Markov chain $\{(S_k,\hat\beta_{k-1}),\ k\geq 0\}$.
\begin{propos}
\label{stronglyirr}
The Markov chain
$\{(S_k,\hat\beta_{k-1}),\ k\geq 0\}$,
with state space $\mathcal S\times\{0,1\}$ and action space $[\lambda_{\min},\lambda_{\max}]$,
 is \emph{strongly irreducible and strongly aperiodic} \cite{Chen}.
\end{propos}
The significance is that,
 if we view the Markov chain as a random walk on a directed graph, this directed graph is always strongly connected and all its states are of period one, irrespective of the input distribution generating $\{(S_k,\hat\beta_{k-1}),\ k\geq 0\}$ \cite{Chen}.
 In particular, under a stationary Markov input distribution,
 whose optimality is proved in \cite{Chen},
  the resulting Markov chain $\{(S_k,\hat\beta_{k-1}),\ k\geq 0\}$ is stationary and ergodic.

The capacity of finite-state Markov channels is studied in~\cite{Chen}, for the case where the encoder 
is provided with the output feedback $\hat\beta_k$. Therein,
 it is proved that, for 
\emph{strongly irreducible and aperiodic} Markov channels, the feedback capacity specialized to our case is given by
\begin{align}
\label{capacity}
C_\Delta^*=\max_{\nu}
\sum_{s=0}^{S_{\max}}\sum_{b\in\{0,1\}}\pi_{\nu}^{(\Delta)}(s,b)I_\nu^{(\Delta)}(s,b),
\end{align}
and is achieved by a stationary Markov input distribution $\nu_\Delta^*$ (the optimizer of (\ref{capacity})), which maps the current state 
$(S_k,\hat\beta_{k-1})=(s,b)$
to a probability distribution over the input signal, $\nu(\lambda|s,b)$.
The terms $\pi_\nu^{(\Delta)}(s,b)$ and $I_\nu^{(\Delta)}(s,b)$ are, respectively, the steady-state distribution
and the instantaneous mutual information rate of state
 $(S_k,\hat\beta_{k-1})=(s,b)$, induced by the input distribution $\nu$.
$I_\nu^{(\Delta)}(s,b)$ is defined as
\begin{align}
\label{ISB}
I_\nu^{(\Delta)}(s,b)=\frac{1}{\Delta}I(\lambda_{k};\hat\beta_k,S_{k+1}|S_k=s,\hat\beta_{k-1}=b,\nu(\cdot|s,b)).
\end{align}
Note that, in general, the input signal distribution $\nu$ is a function of the current state $S_k=s$ and of the feedback signal $\hat\beta_{k-1}=b$.
However, the feedback signal $\hat\beta_{k-1}=b$ is not provided to the encoder in our model.
In the following proposition,
we show that the optimal input distribution solution of (\ref{capacity}), denoted by $\nu_\Delta^*$,
is, in fact, independent of $\hat\beta_{k-1}$.
\begin{propos}
\label{lem1}
The optimal input distribution maximizing the capacity  $C_\Delta^*$ in (\ref{capacity}), $\nu_\Delta^*$, is such that
\begin{align}
\nu_\Delta^*(\cdot|s,0)=\nu_\Delta^*(\cdot|s,1),\ \forall s\in\mathcal S,
\end{align}
\emph{i.e.}, it is independent of the feedback signal $\hat\beta_{k-1}=b$.
Under such a distribution, the capacity is given by
\begin{align}
\label{capacity2}
C_\Delta^*=
\max_{\nu}
\sum_{s=0}^{S_{\max}}\pi_{\nu}^{(\Delta)}(s)I_\nu^{(\Delta)}(s),
\end{align}
where the optimization is over the set of stationary Markov input distributions $\nu$ which map the current cable state 
$S_k=s$
to a probability distribution over the input signal, $\nu(\lambda|s)$.
The terms $\pi_\nu^{(\Delta)}(s)$ and $I_\nu^{(\Delta)}(s)$ are, respectively,
 the steady-state distribution
 and the mutual information rate in state $S_k{=}s$, induced by $\nu$.
$I_\nu^{(\Delta)}(s)$ is defined as
\begin{align}
\label{irate}
I_\nu^{(\Delta)}(s)=\frac{1}{\Delta}I(\lambda_{k};\hat\beta_k,S_{k+1}|S_k=s,\nu(\cdot|s)).
\end{align}
\end{propos}
\begin{proof}
See Appendix A.
\end{proof}

That a distribution $\nu(\cdot|s)$ independent of the feedback $\hat\beta_{k-1}$ is optimal can be intuitively explained:
the state $S_k$ is an \emph{information state}, and captures all information about the past,
so that knowledge of $\hat\beta_{k-1}$ is irrelevant; moreover,
since the instantaneous information rate $I_\nu^{(\Delta)}(S_k,\hat\beta_{k-1})$
 is a concave function of the input distribution, 
 it is maximized by the input distribution $\nu(\cdot|S_k)$ independent of the feedback signal $\hat\beta_{k-1}$,
rather than by randomizing between $\nu(\cdot|S_k,0)$ and $\nu(\cdot|S_k,1)$, depending on the value of $\hat\beta_{k-1}\in\{0,1\}$.
The implication is that the capacity $C_{\Delta}^*$  in Eq. (\ref{capacity}) is
achievable only with the CSI $S_k$ available at the encoder,
but without feedback $\hat\beta_{k-1}$, as it is assumed in our model.

In this paper, we are interested in the analysis of the asymptotic regime $\Delta\to 0$,
thus extending the capacity of \emph{static} Poisson channels \cite{Wyner} to our \emph{dynamic} setting.
The asymptotic capacity, defined as
\begin{align}
C^*\triangleq\lim_{\Delta\to 0}C_\Delta^*,
\end{align}
 is characterized in Proposition~\ref{lem2}.
\begin{propos}
\label{lem2}
The asymptotic capacity $C^*$ is given by
\begin{align}
\label{cap2}
C^*
=
\max_{\nu}
\sum_{s=0}^{S_{\max}}\pi_{\bar\lambda}(s)I_\nu(s),
\end{align}
where we have defined the asymptotic mutual information rate as
\begin{align}
\label{asymrate}
I_\nu(s){\triangleq}\lim_{\Delta\to 0}I_\nu^{(\Delta)}(s)
{=}
A(s)\int_{\lambda_{\min}}^{\lambda_{\max}}\nu(\lambda|s)\lambda
\log_2\left(\frac{\lambda}{\bar\lambda(s)}\right)\mathrm d\lambda,
\end{align}
and the asymptotic steady-state distribution as
\begin{align}
\label{SSD}
\pi_{\bar\lambda}(s)\triangleq \lim_{\Delta\to 0}\pi_\nu^{(\Delta)}(s)
=
\prod_{j=0}^{s-1}\frac{A(j)\bar\lambda(j)}{\mu(j+1)}
\pi_{\bar\lambda}(0),
\end{align}
where
\begin{align}
\label{SSD0}
\pi_{\bar\lambda}(0)\triangleq \lim_{\Delta\to 0}\pi_\nu^{(\Delta)}(0)
=
\frac{1}{1+\sum_{t=1}^{S_{\max}}\prod_{j=0}^{t-1}\frac{A(j)\bar\lambda(j)}{\mu(j+1)}},
\end{align}
and we have defined the average desired input electron intensity
\begin{align}
\label{barlambda}
\bar\lambda(s)
\triangleq
\int_{\lambda_{\min}}^{\lambda_{\max}}\nu(\lambda|s)\lambda\mathrm d\lambda.
\end{align}
\end{propos}
\begin{proof}
See Appendix B.
\end{proof}
Note that the steady-state distribution (\ref{SSD})
is a function of $\bar\lambda$ only. Since we are considering 
infinitely small time-slot durations $\Delta \to 0$,
the input signal $\{\lambda_k\}$ averages out over short time intervals, 
hence only its expected value $\bar\lambda(s)$ 
affects state transitions, rather than its specific distribution.

The capacity optimization problem in Eq. (\ref{cap2}) highlights the following trade-off:
the optimal input signal should, on the one hand,
achieve high instantaneous information rate $I_\nu(S_k)$ (\emph{e.g.}, by employing the myopic distribution
$\nu(\cdot|s)=\arg\max I_\nu(s),\ \forall s$);
on the other hand, it should be designed in such a way as to 
favor the occurrence  of states characterized by large instantaneous information rate.
These two goals are in tension. In fact, the instantaneous information rate is maximum in states with large clogging state $A(s)\simeq 1$, \emph{i.e.},
when $S_k$ is small and the bacterial cable is deprived of electrons.
Visits to these states are achieved more frequently by choosing $\lambda_k=\lambda_{\min}$ with probability one.
However, under this deterministic input distribution the instantaneous information rate is zero.
The optimization problem in (\ref{cap2}) can be interpreted as a MDP,
with state space $\mathcal S$, action $\nu(\cdot|s)$ in each state (each action is a probability distribution over
the input signal $\lambda_{k}\in[\lambda_{\min},\lambda_{\max}]$),
the reward $r(\nu(\cdot|s),s)$ under action $\nu(\cdot|s)$ in state $s$ is given by
\begin{align}
r(\nu(\cdot|s),s)
=
A(s)\int_{\lambda_{\min}}^{\lambda_{\max}}\nu(\lambda|s)\lambda
\log_2\left(\frac{\lambda}{\bar\lambda(s)}\right)\mathrm d\lambda,
\end{align}
and transition probability from state $S_k=s_0$ to $S_{k+1}=s_1$ under action $\nu(\cdot|s_0)$ given by
\begin{align}
\label{txprob}
&\mathbb P(S_{k+1}=s_1|S_k=s_0,\nu(\cdot|s_0))
=
\left\{
\begin{array}{ll}
1-\delta\mu(s_0)-\delta A(s_0)\bar\lambda(s_0) & s_1=s_0,\\
\delta A(s_0)\bar\lambda(s_0)& s_1=s_0+1,\\
\delta\mu(s_0) & s_1=s_0-1,\\
0& \text{otherwise},
\end{array}
\right.
\end{align}
where
$\bar\lambda(s)$ is defined in (\ref{barlambda}) and $\delta$ is any constant satisfying 
\begin{align}
\label{deltacond2}
\delta<\frac{1}{\mu_{\max}+\lambda_{\max}},
\end{align}
in order to guarantee a feasible transition probability matrix.
In fact, by solving the steady-state equations with transition probabilities in Eq. (\ref{txprob}), we obtain
 the steady-state distribution $\pi_{\bar\lambda}(s)$ stated in the proposition.
 Note that the transition probabilities in Eq. (\ref{txprob}) are equivalent to 
 $\mathbb E[p_{S}^{(\delta)}(s_1|s_0,\lambda)]$ (see Eq. (\ref{ps})), where the expectation is with respect to $\lambda\sim \nu(\lambda|s_0)$), but without the terms
 $g_{a,b}(\delta,s_0,\lambda)$, which are negligible   in the asymptotic regime $\Delta\to 0$.

Therefore, the optimal input distribution
which maximizes the capacity in Eq. (\ref{cap2}), denoted by $\nu^*(\cdot|s)$, can be determined by using standard MDP algorithms, such as policy iteration \cite{Bertsekas2005}.
Note that, under any $\nu$, the steady-state distribution of $S_k$ is a function of
the average desired input electron intensity $\bar\lambda$ only. 
 Given $\bar\lambda$, there exist an infinite set of distributions $\nu$
 which induce the same average desired input electron intensity $\bar\lambda$ and the same steady-state distribution $\pi_{\bar\lambda}$.
  Mathematically, this set is defined as
 \begin{align}
 \label{setV}
 \mathcal V(\bar\lambda)\equiv
 \left\{
 \nu:
 \int_{\lambda_{\min}}^{\lambda_{\max}}\nu(\lambda|s)\lambda\mathrm d\lambda=\bar\lambda(s),\ \forall s\in\mathcal S
\right\}.
 \end{align}
 Therefore, we have a degree of freedom in optimizing $\nu$ over such set $\mathcal V(\bar\lambda)$, for every possible choice of $\bar\lambda$.
Since all $\nu\in\mathcal V(\bar\lambda)$ induce the same steady-state distribution,
this optimization is equivalent to maximizing the instantaneous mutual information rate $I_\nu(s)$ in each state.
Intuitively, $I_\nu(s)$ is maximized by choosing input symbols $\lambda_k$ which
can be distinguished more clearly at the receiver and are maximally different, \emph{i.e.},
$\lambda_k\in\{\lambda_{\min},\lambda_{\max}\}$.
This is formalized in the following proposition, which proves the optimality of binary input distributions.
\begin{propos}
\label{lem4}
The optimal input distribution has the form
\begin{align}
\label{muopt}
&\!\!\!\!\nu^*(\lambda_{\max}|s){=}\mathbb P(\lambda_{k}=\lambda_{\max}|S_k=s)
{=}\frac{\bar\lambda^*(s)-\lambda_{\min}}{\lambda_{\max}-\lambda_{\min}},
\\
&\!\!\!\!\nu^*(\lambda_{\min}|s){=}\mathbb P(\lambda_{k}=\lambda_{\min}|S_k=s)
{=}\frac{\lambda_{\max}-\bar\lambda^*(s)}{\lambda_{\max}-\lambda_{\min}},
\end{align}
where $\bar\lambda^*(s)$ is the optimal expected desired input electron intensity in state $s$,
defined as the maximizer of the capacity,
\begin{align}
\label{cap3}
C^*
=
\max_{\bar\lambda:\mathcal S \mapsto [\lambda_{\min},\lambda_{\max}]}
\sum_{s=0}^{S_{\max}}\pi_{\bar\lambda}(s) I(\bar\lambda(s),s),
\end{align}
where $I(x,s)$ is given by, for $x\in[\lambda_{\min},\lambda_{\max}]$,
\begin{align}
\label{imu}
I(x,s)
{\triangleq}
A(s)\left[x\log_2\left(\frac{\lambda_{\max}}{x}\right){-}\frac{\rho(\lambda_{\max}{-}x)}{1{-}\rho}\log_2\left(\frac{1}{\rho}\right)\right].
\end{align}
\end{propos}
\begin{proof}
See Appendix C.
\end{proof}

A similar result has been proved in \cite{Wyner} for the case of a \emph{static} Poisson channel.
Therein, the optimal input distribution is the \emph{myopic} input distribution, 
studied in Section~\ref{MP}, which maximizes the instantaneous mutual information rate, \emph{i.e.},
\begin{align}
\label{eqmp}
\bar\lambda_{MP}\triangleq\underset{x\in[\lambda_{\min},\lambda_{\max}]}{\arg\max}I(x,s),\ \forall s\in\mathcal S.
\end{align}
(Note from (\ref{imu}) that $\bar\lambda_{MP}$ is independent of $s$; the dependence on $s$ has thus been removed accordingly).
In fact, for a static Poisson channel, the tension between maximizing the instantaneous mutual information rate and inducing a favorable 
steady-state distribution does not arise.
Proposition \ref{lem4} thus represents the extension of~\cite{Wyner} to finite-state Markov channels.

From  Proposition~\ref{lem4},
it follows that the optimization can be restricted only to binary input distributions,
which allocate non-zero probability only to the minimum and maximum input electron intensities,
and zero probability to any intermediate values.
Equivalently, only the expected electron intensity $\bar\lambda(s)$ needs to be optimized in each state,
as the maximizer of (\ref{cap3}), using optimization tools such as policy iteration \cite{Bertsekas2005}.

Note that the myopic policy, defined in (\ref{eqmp}) and studied in Section~\ref{MP},
maximizes the instantaneous mutual information rate $I(x,s)$.
However, in doing so, the myopic policy neglects the steady-state behavior of the cable, and may thus induce frequent visits to states characterized by small clogging state
 $A(S_k)\simeq 0$, \emph{i.e.},  when $S_k$ approaches $S_{\max}$.
On the other hand, the optimal input distribution $\bar\lambda^*(s)$ balances this tension by
giving up part of the instantaneous transfer rate in favor of a better steady-state distribution in the future, \emph{i.e.},
states characterized by large clogging state
 $A(S_k)\simeq 1$ where a larger instantaneous mutual information rate can be achieved.
Note that any input distribution $\bar\lambda(s)$ larger than the myopic one $\bar\lambda_{MP}$
is deleterious to the capacity for the following two reasons:
1) a lower instantaneous mutual information rate is achieved, compared to $\bar\lambda_{MP}$ (by definition of the myopic distribution, which maximizes $I(x,s)$);
2) faster recharges of electrons within the cable are induced,
resulting in frequent clogging of the cable, where the instantaneous mutual information rate is small.
This property is formalized in the following
 theorem, which proves a structural property of $\bar\lambda^*(s)$, by comparing it with 
the myopic distribution.
\begin{thm}
\label{thm1}
The optimal signal distribution $\bar\lambda^*$ 
is such that $\bar\lambda^*(s)\leq\bar\lambda_{MP},\forall s$.
\end{thm}
\begin{proof}
See Appendix D.
\end{proof}

Theorem \ref{thm1} can be exploited in the capacity optimization problem (\ref{cap3}),
by restricting the optimization only to distributions such that $\bar\lambda(s)\in[\lambda_{\min},\bar\lambda_{MP}]$.
We now present the policy iteration algorithm \cite{Bertsekas2005}, which exploits this fact.

\begin{algo}[\textbf {Policy iteration algorithm}]
\label{PIA}
\begin{enumerate}
\item \emph{Initialization}: initial average desired electron intensity $\bar \lambda^{[0]}$; $i=0$; $\epsilon>0$
\item \emph{Stage $i$, policy evaluation step}:
compute the achievable rate under policy $\bar \lambda^{[i]}$,
 \begin{align}
C^{[i]}
=
\sum_{s=0}^{S_{\max}}\pi_{\bar\lambda^{[i]}}(s) I(\bar\lambda^{[i]}(s),s),
\end{align}
using (\ref{SSD}) and (\ref{imu}).
Let $D^{[i]}(0)=0$, and recursively, for $s\in\{1,2,\dots,S_{\max}\}$,
 \begin{align}
 \label{Di}
\!\!\!\!\!\!\!\!\!\!D^{[i]}(s)=&
\frac{\mu(s-1)D^{[i]}(s-1)+C^{[i]}-I(\bar\lambda^{[i]}(s-1),s-1)}{A(s-1)\bar\lambda^{[i]}(s-1)}.
\end{align}
\item \emph{Stage $i$, policy improvement step}: for each $s\in\mathcal S$, determine a new policy $\bar\lambda^{[i+1]}$ as follows:
\begin{align}
\label{PIS}
\bar\lambda^{[i+1]}(s)=
\left\{
\begin{array}{ll}
\bar\lambda_{MP} & D^{[i]}(s+1)\geq 0,\\
\lambda_{\min} & D^{[i]}(s+1)\leq \log_2(e)+\frac{1}{1-\rho}\log_2(\rho),\\
\bar\lambda_{MP}2^{D^{[i]}(s+1)} &\text{otherwise}.
\end{array}
\right.
\end{align}
\item \emph{Convergence test}: if $C^{[i+1]}-C^{[i]}<\epsilon$, return $\bar\lambda^{[i+1]}$;
otherwise, let $i:=i+1$ and repeat from step 2).
\end{enumerate}
\end{algo}
The following proposition states the optimality of Algorithm \ref{PIA}.
%{-5mm}
\begin{propos}\label{lem5}
Algorithm \ref{PIA} determines the optimal policy  $\bar\lambda^*$ when $\epsilon\to 0$.
\end{propos}
%{-3mm}
\begin{proof}
The optimality of the policy iteration algorithm is proved in \cite{Bertsekas2005}.
The specific forms of the
 policy evaluation and improvement steps, which exploit the structure of our model, are proved in Appendix~E.%\ref{proofofPIA}.
\end{proof}

The algorithm can be initialized with the myopic policy, studied in Section~\ref{MP}.

\subsection{Myopic codebook generation}
\label{MP}
The myopic policy,
defined in (\ref{eqmp}),
 is the expected desired input electron intensity which maximizes the
instantaneous mutual information rate.
Since the mutual information rate $I(x,s)$ is a concave function of $x$,
with $I^\prime(\lambda_{\min},s)>0$ and $I^\prime(\lambda_{\max},s)<0$,
\begin{comment}
\begin{align}
I^\prime(\lambda_{\min},s)
=
A(s)\log_2\left(\frac{1}{\rho}\right)\frac{\lambda_{\max}}{\lambda_{\max}-\lambda_{\min}}
- A(s)\log_2(e)
>0,
\end{align}
and
\begin{align}
I^\prime(\lambda_{\max},s)
=
- A(s)\log_2(e)
+A(s)\frac{\lambda_{\min}}{\lambda_{\max}-\lambda_{\min}}\log_2\left(\frac{\lambda_{\max}}{\lambda_{\min}}\right).
<0.
\end{align}
Therefore, the mutual information rate $I(x,s)$
\end{comment}
it is  maximized by the unique $x\in(\lambda_{\min},\lambda_{\max})$
such that $I^\prime(x,s)=0$. From (\ref{derivis}), it is given by
\begin{align}
\label{MPexpr}
\bar\lambda_{MP}=\frac{\lambda_{\max}}{e}\rho^{-\frac{\rho}{1-\rho}},
\end{align}
and is independent of $s$ (the argument $s$ has been removed accordingly). The resulting
% mutual information rate is given by
%\begin{align}
%I(\bar\lambda_{MP}(s),s)
%=&
% A(s)\frac{\lambda_{\max}}{e}\left(\frac{\lambda_{\max}}{\lambda_{\min}}\right)^{\frac{\lambda_{\min}}{\lambda_{\max}-\lambda_{\min}}}
% \log_2(e)
%\nonumber\\&
%- A(s)\frac{\lambda_{\min}\lambda_{\max}}{\lambda_{\max}-\lambda_{\min}}\log_2\left(\frac{\lambda_{\max}}{\lambda_{\min}}\right),
%\end{align} and the
  achievable rate under this myopic input distribution, denoted by $C_{MP}$, is given by
\begin{align}
C_{MP}=\bar A_{MP}\lambda_{\max}
\left[
 \frac{1}{e}\rho^{-\frac{\rho}{1-\rho}}
\log_2(e)
+\frac{\rho}{1-\rho}\log_2(\rho)
\right],
\end{align}
where $\bar A_{MP}=\sum_s\pi_{MP}(s)A(s)$ is the average clogging state and
 $\pi_{MP}(s)$ is the steady-state distribution under the myopic policy.
 Note that $\bar\lambda_{MP}$ only maximizes the instantaneous information rate,
without taking into account its effect on the
steady-state distribution, so that the resulting average clogging state $\bar A_{MP}$, and the capacity $C_{MP}$,  may be small.

\section{Numerical Results}
\label{numres}
In this section, we present numerical results.
We consider a cable with electron capacity $S_{\max}=1000$.
The clogging state $A(s)$ and output rate $\mu(s)$ are given by\footnote{The specific choices of $A(s)$ and $\mu(s)$ have been discussed with 
Prof. M. Y. El-Naggar and S. Pirbadian, Department of Physics and Astronomy, University of Southern California, Los Angeles, USA.}
\begin{align}
\label{A}
&A(s)=
\chi(s<S_{\max})\left[
1-(1-A_{\min})\frac{s}{S_{\max}}\right],
\\
&\mu(s)=0.6+0.8\frac{s}{S_{\max}},
\end{align}
where $A_{\min}$ is the \emph{minimum clogging state value}, taking values in the range $[0,1]$.
Note that $A(s)$ is peculiar to the bacterial cable, since it captures 
saturation effects occurring locally within each cell, and causing the overall cable to clog.

 \begin{figure}[t]
\centering
\includegraphics[width = .8\linewidth,trim = 10mm 4mm 10mm 9mm,clip=false]{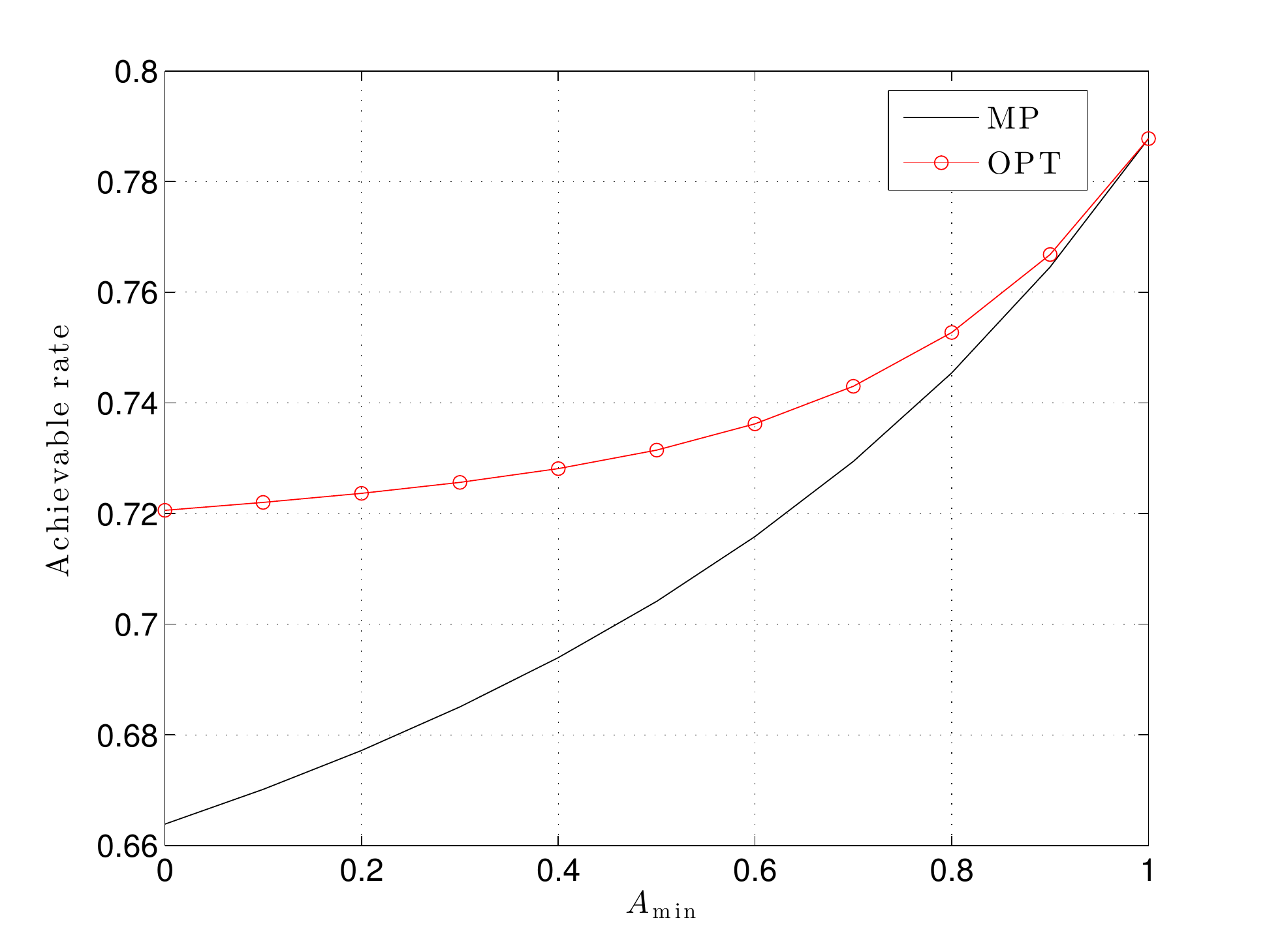}
\caption{Achievable information rate as a function of the minimum clogging state value, $A_{\min}$, for the
MP and OPT input distributions.}%\vspace{-4mm}
\label{fig1}
\end{figure}

In Fig. \ref{fig1}, we evaluate the achievable information rate under the optimal input distribution (OPT),
computed with the policy iteration algorithm \cite{Bertsekas2005},
and the myopic input distribution (MP).
We plot the achievable information rate as a function of $A_{\min}\in[0,1]$.
We note that the achievable information rate increases with $A_{\min}$ for both schemes.
This is because, as $A_{\min}$ increases, both $A(s)$ and the instantaneous information rate $I(x,s)$ increase as well (see Eqs. (\ref{A}) and (\ref{imu})).
Intuitively,  the larger $A(s)$, the better the ability of the  bacterial cable to transport electrons.

OPT outperforms MP by $\sim 9\%$ for small values of $A_{\min}$. This can be explained with the help of Fig. \ref{fig2}, for the case $A_{\min}{=}0$.
In this case, clogging is severe when the cable state approaches the maximum value $S_{\max}$. For instance, 
if $S_k>S_{\max}/2$, then $A(S_{k})<0.5$ and the rate of electrons entering the cable is less than halved, resulting in 
low instantaneous information rate. Therefore, in order to achieve high instantaneous information rate, the
state of the cable $S_k$ should be kept small, \emph{e.g.}, below $S_{\max}/2$.
MP greedily maximizes the instantaneous information rate, but this action results in an unfavorable steady-state distribution, 
such that the cable is often in large queue states $S_k>S_{\max}/2$, where the clogging is significant ($A(S_k)<0.5$) and most electrons are discarded at the cable input.
On the other hand, OPT gives up some instantaneous information rate  in order to favor the occupancy of 
low queue states, where $A(s)$ approaches one and the transfer of information is maximum.
Note that MP does not require CSI at the encoder, as seen in (\ref{MPexpr}). Therefore, Fig. \ref{fig2} demonstrates the 
importance of CSI, which enables the adaptation of the input signal based on the state of the cable.
Finally, note that OPT and MP approach the same value of the average information rate when $A_{\min}{\to}1$. In fact, in this case
the instantaneous information rate $I(x,s)$ is the same in all states, except $S_{\max}$, where $A(S_{\max}){=}0$ and $I(x,S_{\max}){=}0$.
However, state $S_{\max}$ is visited very infrequently, resulting in a negligible degradation of MP with respect to OPT.

 \begin{figure}[t]
\centering
\includegraphics[width = .8\linewidth,trim = 10mm 4mm 10mm 9mm,clip=false]{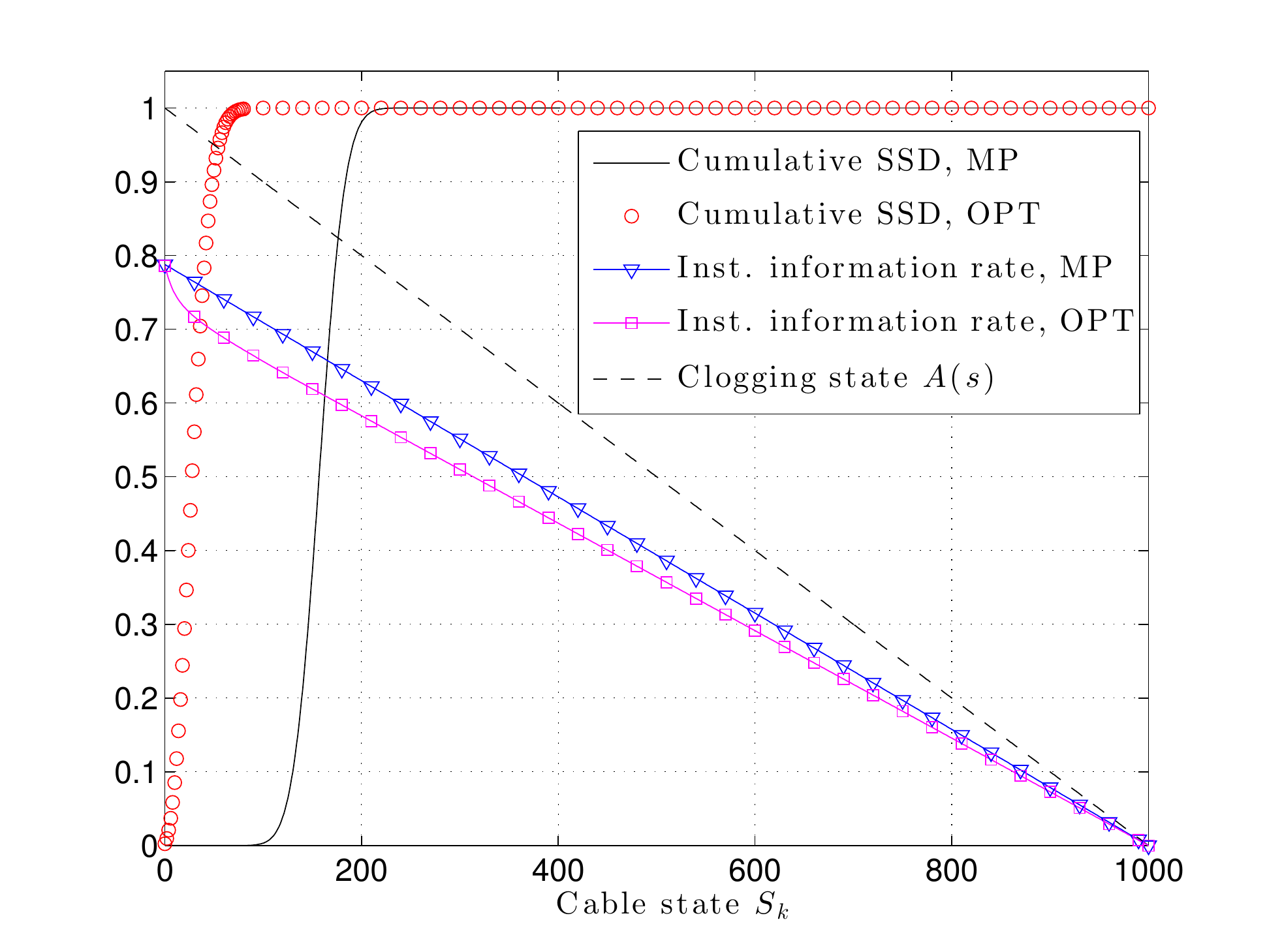}
\caption{Cumulative steady-state distribution (SSD) and 
instantaneous information rate for the MP and OPT input distributions; $A_{\min}=0$.
}%\vspace{-4mm}
\label{fig2}
\end{figure}

\section{Conclusions}
\label{concl}
In this paper, we have studied the capacity of bacterial cables via electron transfer,
for the case where both  the encoder and the decoder have full causal CSI.
We have studied a discrete-time
version of the system, which enables  the computation of an achievable rate for the continuous-time system, based on known results on the capacity of finite-state Markov channels.
We have analyzed the regime of asymptotically small time-slot duration, 
based on which we have established the optimality of binary Markov input distributions, which are functions of the current state only.
We have shown that 
 the optimal distribution optimally balances the tension between optimizing the instantaneous mutual information rate,
and inducing a favorable steady-state distribution, \emph{e.g.}, to states characterized by large clogging state,
and we have proved that it is smaller than that given by the \emph{myopic policy}, which greedily 
maximizes the instantaneous information rate neglecting its effect on the steady-state distribution of the cable.
We have shown that the optimal probability that generates the binary intensity signal as a function of the current cable state 
can be determined efficiently via the policy iteration algorithm.
Our numerical evaluations
 reveal the importance of CSI, which enables adaptation
of the input signal to the state of the cell, and thus motivates further research on the design of more practical schemes, where CSI is only partially available,
and of state estimation techniques.

This work represents a first contribution towards the design of electron signaling schemes in complex microbial structures, \emph{e.g.}, bacterial cables and biofilms,
 where the tension between maximizing the transfer of information and guaranteeing the well-being of the overall community arises,
and thus motivates further research in the design of signaling schemes in
large-scale microbial systems and bio-films, \emph{e.g}, using methods based on statistical physics~\cite{Mitra}.

\section{Acknowledgments}
The authors would like to thank Prof. M. Y. El-Naggar and S. Pirbadian for their 
valuable insights on the biological processes of bacterial cables,
and Prof. M. N. Kenari for her valuable insights on information theoretic aspects related to
microbial communications.

\section*{Appendix A: Proof of Proposition \ref{lem1}}
Let $\nu$ be a stationary Markov input distribution with feedback,
which maps the current state $(S_k,\hat\beta_{k-1})=(s,b)$
to a probability distribution over the input signal, $\nu(\lambda|s,b)$,
but is independent of the past \cite{Chen}.
From Proposition~\ref{stronglyirr},
the state sequence $\{(S_k,\hat\beta_{k-1}),\ k\geq 0\}$ is a strongly irreducible and aperiodic Markov chain, controlled by $\nu$.
Given $(S_{k-1},\hat\beta_{k-2})=(s^\prime,b^\prime)$, which occurs with steady-state probability
$\pi_{\nu}^{(\Delta)}(s^\prime,b^\prime)$, the system moves to state $(S_k,\hat\beta_{k-1})=(s,b)$
with probability 
$\int_{\lambda_{\min}}^{\lambda_{\max}}\nu(\lambda|s^\prime,b^\prime) p_{S,\hat B}^{(\Delta)}(s,b|s^\prime{,}\lambda)\mathrm d\lambda$,
independent of the past.
Therefore, the steady-state distribution of $(S_k,\hat\beta_{k-1})$ induced by the stationary Markov input distribution $\nu$ is
the unique solution of the system of  steady-state equations
\begin{align}
\label{pimu}
\pi_{\nu}^{(\Delta)}(s,b)
=
\sum_{s^\prime,b^\prime}\pi_{\nu}^{(\Delta)}(s^\prime{,}b^\prime)\int_{\lambda_{\min}}^{\lambda_{\max}}\nu(\lambda|s^\prime,b^\prime) p_{S,\hat B}^{(\Delta)}(s,b|s^\prime{,}\lambda)\mathrm d\lambda,\forall s,b.
\end{align}

We now define a new stationary input distribution without feedback, denoted by $\tilde\nu$.
Let, for each $s\in\mathcal S$,
\begin{align}
\label{tilmu}
\tilde\nu(\lambda|s)
=
\frac{\sum_{b}\pi_{\nu}^{(\Delta)}(s,b)\nu(\lambda|s,b)}{\pi_{\nu}^{(\Delta)}(s)},\ \forall \lambda,
\end{align}
where $\pi_{\nu}^{(\Delta)}(s)=\pi_{\nu}^{(\Delta)}(s,0)+\pi_{\nu}^{(\Delta)}(s,1)$.
The steady-state distribution of $(S_k,\hat\beta_{k-1})$ induced by $\tilde\nu$ is the unique solution of the
system of equations
\begin{align}
\label{pitilmu}
\pi_{\tilde\nu}^{(\Delta)}(s,b)
{=}
\sum_{s^\prime,b^\prime}\pi_{\tilde\nu}^{(\Delta)}(s^\prime,b^\prime)\!\!\int_{\lambda_{\min}}^{\lambda_{\max}}\!\!\!\!\tilde\nu(\lambda|s^\prime) p_{S,\hat B}^{(\Delta)}(s,b|s^\prime,\lambda)\mathrm d\lambda.
\end{align}
We now show that $\pi_{\nu}^{(\Delta)}(s,b)$ solves (\ref{pitilmu}), hence $\pi_{\tilde\nu}^{(\Delta)}(s,b)=\pi_{\nu}^{(\Delta)}(s,b),\ \forall s,b$,
\emph{i.e.}, we prove that 
\begin{align}
\label{pitilmu2}
\!\!\!\pi_{\nu}^{(\Delta)}(s,b)
{=}
\sum_{s^\prime,b^\prime}\pi_{\nu}^{(\Delta)}(s^\prime,b^\prime)\!\!\int_{\lambda_{\min}}^{\lambda_{\max}}\!\!\!\!\tilde\nu(\lambda|s^\prime) p_{S,\hat B}^{(\Delta)}(s,b|s^\prime,\lambda)\mathrm d\lambda.
\end{align}
In fact, substituting (\ref{tilmu}) into (\ref{pitilmu2}), we obtain
\begin{align}
&\pi_{\nu}^{(\Delta)}(s,b)
=
\sum_{s^\prime,b^\prime}
\pi_{\nu}^{(\Delta)}(s^\prime,b^\prime)\int_{\lambda_{\min}}^{\lambda_{\max}}
\frac{\sum_{y}\pi_{\nu}(s^\prime,y)\nu(\lambda|s^\prime,y)}{\sum_{y}\pi_{\nu}(s^\prime,y)}
p_{S,\hat B}^{(\Delta)}(s,b|s^\prime,\lambda)\mathrm d\lambda
\nonumber\\&
=
\sum_{s^\prime,y}
\pi_{\nu}^{(\Delta)}(s^\prime,y)\int_{\lambda_{\min}}^{\lambda_{\max}}\nu(\lambda|s^\prime,y)p_{S,\hat B}^{(\Delta)}(s,b|s^\prime,\lambda)\mathrm d\lambda,
\end{align}
and this holds true from (\ref{pimu}).
Therefore, the two policies $\nu$ and $\tilde\nu$ achieve the same steady-state distribution
$\pi_{\nu}^{(\Delta)}(s,b)=\pi_{\tilde\nu}^{(\Delta)}(s,b),\ \forall s, b$.

The mutual information rate $I_\nu^{(\Delta)}(s,b)$ is given by
\begin{align}
\label{Isb2}
I_\nu^{(\Delta)}(s,b)=
\sum_{s^\prime,b^\prime}Z^{(\Delta)}(\nu(\cdot|s,b),s,s^\prime,b^\prime),
\end{align}
where from  (\ref{ISB}) we have defined
\begin{align}
\label{zdelta}
&\!\!\!Z^{(\Delta)}(x,s,s^\prime,b^\prime)=
-\frac{1}{\Delta}\int_{\lambda_{\min}}^{\lambda_{\max}}x(\lambda)p_{S,\hat B}^{(\Delta)}(s^\prime,b^\prime|s,\lambda)
\log_2\left(\frac{\int_{\lambda_{\min}}^{\lambda_{\max}}x(\lambda^\prime)p_{S,\hat B}^{(\Delta)}(s^\prime,b^\prime|s,\lambda^\prime)
\mathrm d\lambda^\prime
}{p_{S,\hat B}^{(\Delta)}(s^\prime,b^\prime|s,\lambda)}\right)\mathrm d\lambda.
\end{align}
It can be shown that  $Z^{(\Delta)}(x,s,s^\prime,b^\prime)$ is a strictly concave function of the input distribution $x$.
Therefore,  using (\ref{Isb2}) we obtain
\begin{align}
&
\sum_{b\in\{0,1\}}
\pi_\nu^{(\Delta)}(s,b)I_\nu^{(\Delta)}(s,b)
=\pi_\nu^{(\Delta)}(s)\sum_{s^\prime,b^\prime}
\sum_{b\in\{0,1\}}
\frac{\pi_\nu^{(\Delta)}(s,b)}{\pi_\nu^{(\Delta)}(s)}Z^{(\Delta)}(\nu(\cdot|s,b),s,s^\prime,b^\prime)
\nonumber\\&
\leq
\pi_\nu^{(\Delta)}(s)
\sum_{s^\prime,b^\prime}
Z^{(\Delta)}\left(
\sum_{b\in\{0,1\}}
\frac{\pi_\nu^{(\Delta)}(s,b)}{\pi_\nu^{(\Delta)}(s)}\nu(\cdot|s,b)
,s,s^\prime,b^\prime\right)
\nonumber\\&
=
\pi_{\tilde\nu}^{(\Delta)}(s)
\sum_{s^\prime,b^\prime}Z^{(\Delta)}(\tilde\nu(\cdot|s),s,s^\prime,b^\prime)
=\pi_{\tilde\nu}^{(\Delta)}(s)I_{\tilde\nu}^{(\Delta)}(s),
\nonumber
\end{align}
with equality if and only if $\nu(\cdot|s,b)=\tilde \nu(\cdot|s), \forall b\in\{0,1\}$,
where $I_{\tilde\nu}^{(\Delta)}(s)$ is given by (\ref{irate}).
 In the last step, we have used the definition of $\tilde\nu$ in (\ref{tilmu}) and the fact that $\pi_{\nu}^{(\Delta)}(s)=\pi_{\tilde\nu}^{(\Delta)}(s)$.
Finally, we obtain
\begin{align}
\sum_{s=0}^{S_{\max}}
\sum_{b\in\{0,1\}}
\pi_{\nu}(s,b)I_\nu(s,b)
=
\sum_{s=0}^{S_{\max}}
\sum_{b\in\{0,1\}}
\pi_{\tilde\nu}(s,b)I_\nu(s,b)
\leq
\sum_{s=0}^{S_{\max}}\pi_{\tilde\nu}(s)I_{\tilde\nu}(s),
\end{align}
with equality if and only if $\tilde\nu(\lambda|s)=\nu(\lambda|s,b),\forall \lambda,s,b$,
so that the input distribution $\tilde\nu$ achieves a larger average information rate than the original input distribution $\nu$.
The proposition is thus proved.

\section*{Appendix B: Proof of Proposition \ref{lem2}}
The mutual information rate $I_\nu^{(\Delta)}(s)$ is 
defined in (\ref{irate}) and is given by
\begin{align}
\label{48}
I_\nu^{(\Delta)}(s)=
\sum_{s^\prime,b^\prime}Z^{(\Delta)}(\nu(\cdot|s),s,s^\prime,b^\prime),
\end{align}
where $Z^{(\Delta)}(\nu(\cdot|s),s,s^\prime,b^\prime)$ is defined in (\ref{zdelta}).
In particular, for $(s^\prime,b^\prime)\notin\{(s,0),(s+1,0),(s-1,1)\}$,
from (\ref{Pab}) 
when $\Delta\to 0$, we have that 
\begin{align}
\lim_{\Delta\to 0}Z^{(\Delta)}(\nu(\cdot|s),s,s^\prime,b^\prime)=0,
\end{align}
since, in this case, $p_{S,\hat B}^{(\Delta)}(s^\prime,b^\prime|s,\lambda)\sim o(\Delta^2)$.
Therefore, it follows that
\begin{align}
&\lim_{\Delta\to 0}
I_\nu^{(\Delta)}(s)
=
\lim_{\Delta\to 0}
\left[
Z^{(\Delta)}(\nu(\cdot|s),s,s,0)
\right.\nonumber\\&\left.+
Z^{(\Delta)}(\nu(\cdot|s),s,s+1,0)+
Z^{(\Delta)}(\nu(\cdot|s),s,s-1,1)
\right].
\end{align}
By neglecting the probability terms of order $o(\Delta^2)$ in (\ref{ps})-(\ref{11}),
from (\ref{zdelta}) we thus obtain
\begin{align}
&Z^{(\Delta)}(\nu(\cdot|s),s,s,0)
{\simeq}
{-}\!\!\!\int_{\lambda_{\min}}^{\lambda_{\max}}\!\!\!\!\!\!\!\!\!x(\lambda)\left(\frac{1}{\Delta}{-}\mu(s){-}A(s)\lambda\right)
%\nonumber\\&\qquad\times
\log_2\left(\frac{1/\Delta-\mu(s)- A(s)\bar\lambda(s)}{1/\Delta-\mu(s)-A(s)\lambda}\right)\mathrm d\lambda,
\nonumber\\&\qquad
\to
-\log_2(e)\int_{\lambda_{\min}}^{\lambda_{\max}}x(\lambda)A(s)(\lambda-\bar\lambda(s))
     \mathrm d\lambda=0,\label{sdfsdfsdfsdf}
\\&
Z^{(\Delta)}(x,s,s+1,0)
{=}
\int_{\lambda_{\min}}^{\lambda_{\max}}\!\!\!\!\!x(\lambda) A(s)\lambda\log_2\left(\frac{\lambda}{\bar\lambda(s)}\right)\mathrm d\lambda,
\nonumber\\&
Z^{(\Delta)}(\nu(\cdot|s),s,s-1,1)\simeq 0,
\end{align}
where, in (\ref{sdfsdfsdfsdf}), we have applied L'Hopital's rule to compute the limit.
Finally, using (\ref{48}), we obtain $\lim_{\Delta\to 0}I_\nu^{(\Delta)}(s)=I_\nu(s)$
as given by (\ref{asymrate}).

We now compute the asymptotic steady-state distribution. 
The steady-state distribution $\pi_{\nu}^{(\Delta)}$ is the unique solution of
the balance equations, for $s\in\{1,2,\dots,S_{\max}\}$,
\begin{align}
&\sum_{s^\prime=0}^{s-1}\pi_{\nu}^{(\Delta)}(s^\prime)
\sum_{s^{\prime\prime}=s}^{S_{\max}}\int_{\lambda_{\min}}^{\lambda_{\max}}\nu(\lambda|s^\prime)p_{S}^{(\Delta)}(s^{\prime\prime}|s^\prime,\lambda)\mathrm d\lambda
\nonumber\\&
=
\sum_{s^{\prime}=s}^{S_{\max}}\pi_{\nu}^{(\Delta)}(s^\prime)
\sum_{s^{\prime\prime}=0}^{s-1}
\int_{\lambda_{\min}}^{\lambda_{\max}}\nu(\lambda|s^\prime)p_{S}^{(\Delta)}(s^{\prime\prime}|s^\prime,\lambda)\mathrm d\lambda.
\end{align}
By reordering the terms, we can rewrite
\begin{align}
%\
\label{above}
&\pi_{\nu}^{(\Delta)}(s-1)
\int_{\lambda_{\min}}^{\lambda_{\max}}\!\!\!\nu(\lambda|s-1)p_{S}^{(\Delta)}(s|s-1,\lambda)\mathrm d\lambda
%\nonumber\\&
-\pi_{\nu}^{(\Delta)}(s)
\int_{\lambda_{\min}}^{\lambda_{\max}}\!\!\!\nu(\lambda|s)p_{S}^{(\Delta)}(s-1|s,\lambda)\mathrm d\lambda
\\&
=
\pi_{\nu}^{(\Delta)}(s)
\sum_{s^{\prime\prime}=0}^{s-2}
\int_{\lambda_{\min}}^{\lambda_{\max}}\nu(\lambda|s)p_{S}^{(\Delta)}(s^{\prime\prime}|s,\lambda)\mathrm d\lambda
%\nonumber\\&
%
-\pi_{\nu}^{(\Delta)}(s-1)
\sum_{s^{\prime\prime}=s+1}^{S_{\max}}\int_{\lambda_{\min}}^{\lambda_{\max}}\!\!\!\!\nu(\lambda|s^\prime)p_{S}^{(\Delta)}(s^{\prime\prime}|s^\prime,\lambda)\mathrm d\lambda
\nonumber\\&
+\!\!\sum_{s^{\prime}=s+1}^{S_{\max}}\pi_{\nu}^{(\Delta)}(s^\prime)
\sum_{s^{\prime\prime}=0}^{s-1}
\int_{\lambda_{\min}}^{\lambda_{\max}}\!\!\!\!\!\!\nu(\lambda|s^\prime)p_{S}^{(\Delta)}(s^{\prime\prime}|s^\prime,\lambda)\mathrm d\lambda
%
%\nonumber\\&
-\sum_{s^\prime=0}^{s-2}\pi_{\nu}^{(\Delta)}(s^\prime)
\sum_{s^{\prime\prime}=s}^{S_{\max}}\int_{\lambda_{\min}}^{\lambda_{\max}}\!\!\!\!\!\!\nu(\lambda|s^\prime)p_{S}^{(\Delta)}(s^{\prime\prime}|s^\prime,\lambda)\mathrm d\lambda.
\nonumber
\end{align}
Using (\ref{ps})-(\ref{11}) and the fact that the terms $g_{a,b}(\Delta,s,\lambda)$ are of the order of $o(\Delta^{\max\{a+b,2\}})$,
it can be shown that 
the right hand-side of (\ref{above}) contains terms of the order of $o(\Delta^2)$. Therefore,
using (\ref{ps})-(\ref{10}) to express the left-hand side and including the terms $o(\Delta^2)$ in the right-hand side,
(\ref{above}) becomes
\begin{align}
&\pi_{\nu}^{(\Delta)}(s-1)
\int_{\lambda_{\min}}^{\lambda_{\max}}\nu(\lambda|s-1)\Delta A(s-1)\lambda\mathrm d\lambda
%\nonumber\\&
-\pi_{\nu}^{(\Delta)}(s)\Delta\mu(s)
=
o(\Delta^2),
\end{align}
or equivalently, by dividing each side by $\Delta$,
\begin{align}
&\pi_{\nu}^{(\Delta)}(s-1)A(s-1)\bar\lambda(s-1)
-\pi_{\nu}^{(\Delta)}(s)\mu(s)
=
o(\Delta).
\end{align}
Therefore, letting $\Delta\to 0$, we obtain
\begin{align}
\pi_{\bar\lambda}(s)
=
\frac{A(s-1)\bar\lambda(s-1)}{\mu(s)}\pi_{\bar\lambda}(s-1).
\end{align}
The steady-state distribution (\ref{SSD}) is obtained by solving recursively, thus proving
the proposition.

\section*{Appendix C: Proof of Proposition \ref{lem4}}
Using (\ref{setV}), the capacity $C^*$ can be decoupled as
 \begin{align}
C^*
=
\max_{\bar\lambda}
\max_{\nu\in\mathcal V(\bar\lambda)}
\sum_{s=0}^{S_{\max}}\pi_{\bar\lambda}(s)I_\nu(s),
\end{align}
where the inner optimization is over the set of distribution with a predefined 
average desired input electron intensity profile $\bar\lambda$, 
and the outer optimization maximizes over such $\bar\lambda$. 

We now solve the inner optimization problem
 \begin{align}
\max_{\nu\in\mathcal V(\bar\lambda)}
\sum_{s=0}^{S_{\max}}\pi_{\bar\lambda}(s)I_\nu(s).
\end{align}
Since $I_\nu(s)$ is only a function of $\nu(\cdot|s)$ and is independent of $\nu(\cdot|s^\prime),\ \forall s^\prime\neq s$,
and the steady-state distribution $\pi_{\bar\lambda}$ is the same for all $\nu\in\mathcal V(\bar\lambda)$,
we obtain
 \begin{align}
\max_{\nu\in\mathcal V(\bar\lambda)}
\sum_{s=0}^{S_{\max}}\pi_{\bar\lambda}(s)I_\nu(s)
=
\sum_{s=0}^{S_{\max}}\pi_{\bar\lambda}(s)
I(\bar\lambda(s),s),
\end{align}
where we have defined
\begin{align}
I(\bar\lambda(s),s)
=\max_{\nu(\cdot|s)} I_\nu(s),\text{ s.t. }
 \int_{\lambda_{\min}}^{\lambda_{\max}}\nu(\lambda|s)\lambda\mathrm d\lambda=\bar\lambda(s).
\end{align}
Equivalently, using (\ref{asymrate}),
\begin{align}
\label{zzz}
&I(\bar\lambda(s),s)
=\max_{x}
A(s) \int_{\lambda_{\min}}^{\lambda_{\max}}x(\lambda)\lambda\log_2\left(\frac{\lambda}{\bar\lambda(s)}\right)\mathrm d\lambda,
\nonumber\\&
\text{ s.t. } \int_{\lambda_{\min}}^{\lambda_{\max}}x(\lambda)\lambda\mathrm d\lambda=\bar\lambda(s),
\ \int_{\lambda_{\min}}^{\lambda_{\max}}x(\lambda)\mathrm d\lambda=1.
\end{align}
Letting
\begin{align}
z(y)\triangleq y\log_2\left(\frac{y}{\bar\lambda(s)}\right),
\end{align}
the optimization problem (\ref{zzz}) is equivalent to
\begin{align}
&I_\nu^*(s)
=
\max_{x}
A(\mathbf s)\int_{\lambda_{\min}}^{\lambda_{\max}}x(\lambda)z(\lambda)\mathrm d\lambda
\nonumber\\&
\text{ s.t. }\int_{\lambda_{\min}}^{\lambda_{\max}}x(\lambda)\lambda\mathrm d\lambda=\bar\lambda(s),
\ \int_{\lambda_{\min}}^{\lambda_{\max}}x(\lambda)\mathrm d\lambda=1.
\end{align}
Note that  $z(y)$ is a convex function of $y$. 
Therefore, we have that 
\begin{align*}
\int_{\lambda_{\min}}^{\lambda_{\max}}\!\!\!\!\!x(\lambda)z(\lambda)\mathrm d\lambda
{\leq}
\frac{\lambda_{\max}{-}\bar\lambda(s)}{\lambda_{\max}{-}\lambda_{\min}}z(\lambda_{\min})
{+}
\frac{\bar\lambda(s){-}\lambda_{\min}}{\lambda_{\max}{-}\lambda_{\min}}z(\lambda_{\max}).
\end{align*}
Therefore, the maximum in (\ref{zzz}) is attained by a distribution $\nu^*$
which selects $\lambda_{k}\in\{\lambda_{\min},\lambda_{\max}\}$ with  probabilities
given by (\ref{muopt}), so as to attain the constraint $\mathbb E[\lambda_k|S_k=s]=\bar\lambda(s),\forall s$.
Under the optimal distribution,  the expression of 
 $I(x,s)$ can be shown to be as in (\ref{imu}).
The proposition is thus proved.

\section*{Appendix D: Proof of Theorem \ref{thm1}}
We prove the theorem by analyzing structural properties of the $n$-step cost-to-go function \cite{Bertsekas2005} from state $s$, denoted by
 $V_n(s)$.
 Using (\ref{txprob}) and the structural properties of Proposition \ref{lem4}, which restricts the input signal to binary values $\lambda_k\in\{\lambda_{\min},\lambda_{\max}\}$,
$V_n(s)$ solves recursively
\begin{align}
\label{vn}
V_n(s)=\max_{x\in[\lambda_{\min},\lambda_{\max}]}&\left\{I(x,s)+\delta A(s)xV_{n-1}(s+1)+\delta\mu(s)V_{n-1}(s-1)\right.
\nonumber\\&\left.
+(1-\delta  A(s)x-\delta\mu(s))V_{n-1}(s)\right\},
\end{align}
where $V_0(s)=0,\ \forall s$.
The maximizer of (\ref{vn}) is the optimal 
average desired input electron intensity in step $n$, denoted by $\bar\lambda_n(s)$, yielding
\begin{align}
\label{newvn}
V_n(s)=&
V_{n-1}(s)+I(\bar\lambda_n(s),s)
- \delta A(s)\bar\lambda_n(s)[V_{n-1}(s)-V_{n-1}(s+1)]
\nonumber\\&
+\delta\mu(s)[V_{n-1}(s-1)-V_{n-1}(s)].
\end{align}
We prove that
\begin{align}
\label{toprove}
\lambda_{\min}\leq\bar\lambda_n(s)\leq\bar\lambda_{MP},\ \forall s,\ \forall n.
\end{align}
Then, since $\bar\lambda_n(s)$ converges to the optimal policy as $n\to\infty$, \emph{i.e.},
 $\bar\lambda_n(s)\to\bar\lambda^*(s)$ for $n\to\infty$ \cite{Bertsekas2005},
 it follows that the condition (\ref{toprove}) implies $\lambda_{\min}\leq\bar\lambda^*(s)\leq\bar\lambda_{MP}$ by taking the limit $n\to\infty$, thus proving the theorem.

It can be shown that $I(x,s)$, given by (\ref{imu}), is a strictly concave function of $x$, with first derivative
\begin{align}
\label{derivis}
&I^\prime(x,s)
\triangleq
\frac{\mathrm d I(x,s)}{\mathrm d x}
=
A(s)\left[\log_2\left(\frac{\lambda_{\max}}{ex}\right)
+\frac{\rho}{1-\rho}\log_2\left(\frac{1}{\rho}\right)\right].
\end{align}
Therefore,
given $V_{n-1}$, the objective function in (\ref{vn}) is a concave function of $x$. Its derivative with respect to $x$ is given by
\begin{align}
f(x)\triangleq I^\prime(x,s)-\delta A(s)[V_{n-1}(s)-V_{n-1}(s+1)],
\end{align}
We have the following cases:

\noindent{\bf1}) If $f(\lambda_{\max})>0$, or equivalently, by using (\ref{derivis}) and (\ref{MPexpr}),
 $2^{-\delta[V_{n-1}(s)-V_{n-1}(s+1)]}\bar\lambda_{MP}>\lambda_{\max}$,
then $\bar\lambda_n(s)=\lambda_{\max}$.
 It follows that the condition 
$\bar\lambda_n(s)\leq\bar\lambda_{MP}$ is violated.
Note that, for this case to occur, $V_{n-1}(s+1)-V_{n-1}(s)>0$ must necessarily hold.

\noindent{\bf2})
If $f(\lambda_{\min})>0$, or equivalently $2^{\delta[V_{n-1}(s+1)-V_{n-1}(s)]}\bar\lambda_{MP}<\lambda_{\min}$, then
$\bar\lambda_n(s)=\lambda_{\min}$.
 It follows that the condition 
$\bar\lambda_n(s)\leq\bar\lambda_{MP}$ is satisfied. Using (\ref{MPexpr}),
this case is equivalent to
\begin{align}
V_{n-1}(s)-V_{n-1}(s+1)>\frac{\frac{1}{1-\rho}\ln\left(\frac{1}{\rho}\right)-1}{\delta}\log_2(e)\triangleq \gamma.
\end{align}

\noindent{\bf3}) Otherwise, $\bar\lambda_n(s)$ is the unique solution of $f(x)=0$. Using (\ref{derivis}), we thus obtain
\begin{align}
\label{optimasdf}
\bar\lambda_n(s)=2^{-\delta[V_{n-1}(s)-V_{n-1}(s+1)]}\bar\lambda_{MP}.
\end{align}
In this case, the condition 
$\bar\lambda_n(s)\leq\bar\lambda_{MP}$ is satisfied if and only if 
$V_{n-1}(s+1)-V_{n-1}(s)\leq 0$.

By combining cases {\bf1}, {\bf2} and {\bf3}, we obtain that
the condition stated in (\ref{toprove}) is equivalent to 
\begin{align}
\label{sdfsdfsdf}
V_{n-1}(s)-V_{n-1}(s+1)\geq 0.
\end{align}
We prove (\ref{sdfsdfsdf}) by induction on $n$, thus implying (\ref{toprove}). Trivially, (\ref{sdfsdfsdf}) holds for $n=1$, since  $V_{0}(s)=0, \forall s$.
Now, let $n\geq1$ and assume that the condition (\ref{sdfsdfsdf}) holds under such $n$. We show that 
(\ref{sdfsdfsdf}) implies 
\begin{align}
\label{sdfsdfsdf2}
V_{n}(s)-V_{n}(s+1)\geq 0,
\end{align}
proving the induction step.
Since case {\bf1} cannot hold under the induction hypothesis,
by combining cases {\bf2} and {\bf3},
  $\bar\lambda_n(s)$ can be expressed as
\begin{align}
\label{c1}
&\bar\lambda_n(s)=\max\left\{\lambda_{\min},2^{-\delta[V_{n-1}(s)-V_{n-1}(s+1)]}\bar\lambda_{MP}\right\}.
\end{align}
Therefore, by reordering the terms and using (\ref{newvn}), we obtain
\begin{align}
&V_n(s)-V_n(s+1)=I(\bar\lambda_n(s),s)-I(\bar\lambda_n(s+1),s+1)
\nonumber\\&
+
 (V_{n-1}(s)-V_{n-1}(s+1))(1-\delta\mu(s+1)-\delta A(s)\bar\lambda_n(s))
\nonumber\\&
+\delta\mu(s)(V_{n-1}(s-1)-V_{n-1}(s))
+\delta A(s+1)\bar\lambda_n(s+1)(V_{n-1}(s+1)-V_{n-1}(s+2))
\nn\\&
\triangleq g(V_{n-1}(s-1),V_{n-1}(s),V_{n-1}(s+1),V_{n-1}(s+2)).
\end{align}
Therefore, the induction step $V_n(s)-V_n(s+1)\geq 0$ is equivalent to $g(V_{n-1}(s-1),V_{n-1}(s),V_{n-1}(s+1),V_{n-1}(s+2))\geq 0$.
Note that $g(V_{n-1}(s-1),V_{n-1}(s),V_{n-1}(s+1),V_{n-1}(s+2))$ is a non-decreasing function of $V_{n-1}(s-1)$.
Since $V_{n-1}(s-1)\geq V_{n-1}(s)$ from (\ref{sdfsdfsdf}), we thus obtain
\begin{align}
&V_n(s)-V_n(s+1)
= g(V_{n-1}(s-1),V_{n-1}(s),V_{n-1}(s+1),V_{n-1}(s+2))
\nn\\&
\geq 
g(V_{n-1}(s),V_{n-1}(s),V_{n-1}(s+1),V_{n-1}(s+2))
=
I(\bar\lambda_n(s),s)-I(\bar\lambda_n(s+1),s+1)
\nn\\&
+ (V_{n-1}(s)-V_{n-1}(s+1))(1-\delta\mu(s+1)-\delta A(s)\bar\lambda_n(s))
\nonumber\\&
+\delta A(s+1)\bar\lambda_n(s+1)(V_{n-1}(s+1)-V_{n-1}(s+2)).
\label{mona}
\end{align}
We now minimize $g(V_{n-1}(s),V_{n-1}(s),V_{n-1}(s+1),V_{n-1}(s+2))$ with respect to 
$V_{n-1}(s)$, and use the induction hypothesis $V_{n-1}(s)\geq V_{n-1}(s+1)$.
We distinguish the two cases $V_{n-1}(s)- V_{n-1}(s+1)>\gamma $ and $\gamma\geq V_{n-1}(s)- V_{n-1}(s+1)\geq 0$ below.

\paragraph{Case $V_{n-1}(s)- V_{n-1}(s+1)>\gamma $}
In this case, 
(\ref{c1}) yields $\bar\lambda_n(s)=\lambda_{\min}$, and therefore,
substituting in (\ref{mona}),
\begin{align}
&g(V_{n-1}(s),V_{n-1}(s),V_{n-1}(s+1),V_{n-1}(s+2))
\nn\\&
=
 (V_{n-1}(s)-V_{n-1}(s+1))(1-\delta\mu(s+1)-\delta A(s)\lambda_{\min})-I(\bar\lambda_n(s+1),s+1)
\nonumber\\&
+\delta A(s+1)\bar\lambda_n(s+1)(V_{n-1}(s+1)-V_{n-1}(s+2)).
\end{align}
The derivative of $g(V_{n-1}(s),V_{n-1}(s),V_{n-1}(s+1),V_{n-1}(s+2))$ with respect to $V_{n-1}(s)$ is given~by
\begin{align}
\left.\frac{\mathrm dg(v,v,V_{n-1}(s+1),V_{n-1}(s+2))}{\mathrm dv}\right|_{v=V_{n-1}(s)}
=
1-\delta\mu(s+1)-\delta A(s)\lambda_{\min}\geq 0,
\end{align}
where we have used  (\ref{deltacond2}).

\paragraph{Case $\gamma\geq V_{n-1}(s)- V_{n-1}(s+1)\geq0$}
In this case, (\ref{c1}) yields $\bar\lambda_n(s)=2^{-\delta[V_{n-1}(s)-V_{n-1}(s+1)]}\bar\lambda_{MP}$.
Substituting this expression in (\ref{mona}),
 we obtain
\begin{align}
&\left.\frac{\mathrm dg(v,v,V_{n-1}(s+1),V_{n-1}(s+2))}{\mathrm dv}\right|_{v=V_{n-1}(s)}
=
1-\delta\mu(s+1)-\delta A(s)\bar\lambda_n(s)
\nn\\&
+\delta^2 A(s)2^{-\delta[V_{n-1}(s)-V_{n-1}(s+1)]}\bar\lambda_{MP}\ln2(V_{n-1}(s)-V_{n-1}(s+1))^2
\nn\\&
 -I^\prime(\bar\lambda_n(s),s)
2^{-\delta[V_{n-1}(s)-V_{n-1}(s+1)]}\bar\lambda_{MP}
\ln2\delta[V_{n-1}(s)-V_{n-1}(s+1)].
\end{align}
Now, from (\ref{derivis}) and (\ref{MPexpr}),
 we have that 
\begin{align}
I^\prime(\bar\lambda_n(s),s)
=
A(s)\delta[V_{n-1}(s)-V_{n-1}(s+1)],
\end{align}
and therefore
\begin{align}
\left.\frac{\mathrm dg(v,v,V_{n-1}(s+1),V_{n-1}(s+2))}{\mathrm dv}\right|_{v=V_{n-1}(s)}
=
1-\delta\mu(s+1)-\delta A(s)\bar\lambda_n(s)
\geq 0.
\end{align}

In both cases $V_{n-1}(s)- V_{n-1}(s+1)>\gamma $ and $\gamma\geq V_{n-1}(s)- V_{n-1}(s+1)\geq0$,
 we have that $g(V_{n-1}(s),V_{n-1}(s),V_{n-1}(s+1),V_{n-1}(s+2))$ is a non-decreasing function of 
$V_{n-1}(s)$, minimized by $V_{n-1}(s)=V_{n-1}(s+1)$, yielding $\bar\lambda_n(s)=\bar\lambda_{MP}$.
From (\ref{mona}),
we thus obtain
\begin{align}
&V_n(s)-V_n(s+1)
= g(V_{n-1}(s-1),V_{n-1}(s),V_{n-1}(s+1),V_{n-1}(s+2))
\nn\\&
\geq
g(V_{n-1}(s+1),V_{n-1}(s+1),V_{n-1}(s+1),V_{n-1}(s+2))
\nn\\&
=
I(\bar\lambda_{MP},s)-I(\bar\lambda_n(s+1),s+1)
+\delta A(s+1)\bar\lambda_n(s+1)(V_{n-1}(s+1)-V_{n-1}(s+2))
\nn\\&
\geq
I(\bar\lambda_{MP},s)-I(\bar\lambda_n(s+1),s+1),
\end{align}
where in the last inequality we have used the induction hypothesis $V_{n-1}(s+1)-V_{n-1}(s+2)\geq 0$.
Finally, using (\ref{imu}) and the fact that $A(0)=1$ from Assumption \ref{ass1}, we obtain
\begin{align}
&I(\bar\lambda_{MP},s)-I(\bar\lambda_n(s+1),s+1)=A(s)I(\bar\lambda_{MP},0)-A(s+1)I(\bar\lambda_n(s+1),0)
\nn\\&
=
(A(s)-A(s+1))I(\bar\lambda_{MP},0)
+A(s+1)\left[I(\bar\lambda_{MP},0)-I(\bar\lambda_n(s+1),0)\right]
\geq 0,
\end{align}
where we have used the definition of $\bar\lambda_{MP}$ in (\ref{eqmp}) and the fact that 
$A(s)\geq A(s+1)$ from Assumption \ref{ass1}.
Therefore, $V_n(s)-V_n(s+1)\geq 0$, proving the induction step and the theorem.

\section*{Appendix E: Proof of Proposition \ref{lem5}}
In the policy evaluation step \cite{Bertsekas2005}, given
the policy $\bar \lambda^{[i]}$ at the beginning of the $i$th
 stage of the algorithm, we evaluate the value function $v^{[i]}(s)$ over the states $s\in\mathcal S$. 
  %\footnote{In the following proof, we neglect the stage index $i$.}
 Let $C^{[i]}$ be the achievable rate under policy $\bar \lambda^{[i]}$.
 The value function under $\bar \lambda^{[i]}$ is the solution of \cite{Bertsekas2005}
 \begin{align}
v^{[i]}(s){-}\sum_{s_1\in\mathcal S}\mathbb P(S_{k+1}{=}s_1|S_k{=}s,\bar\lambda^{[i]}(s))v^{[i]}(s_1)
 {=}
I(\bar\lambda^{[i]}(s),s)
-C^{[i]},\ \forall s\in\mathcal S
\end{align}
with $v^{[i]}(0)=0$.
Equivalently, using (\ref{txprob}), we obtain
 \begin{align}
&v^{[i]}(0)=0,
\nonumber\\&
v^{[i]}(s)
-\delta\mu(s)v^{[i]}(s-1)
-\delta A(s)\bar\lambda^{[i]}(s)v^{[i]}(s+1)
\\&
-[1-\delta\mu(s)-\delta A(s)\bar\lambda^{[i]}(s)]v^{[i]}(s)
 =
I(\bar\lambda^{[i]}(s),s)
-C^{[i]},\ \forall s\in\mathcal S,
\nonumber
 \label{second}
\end{align}
Let $D^{[i]}(0)=0$ and $D^{[i]}(s)=\delta[v^{[i]}(s)-v^{[i]}(s-1)],\ \forall s>0$.
Then, the recursive expression (\ref{Di}) directly follows by 
rewriting (\ref{second}) in terms of $D^{[i]}$.

In the policy improvement step, an improved policy $\bar\lambda^{[i+1]}(s)$ is determined as
 \begin{align}
\bar\lambda^{[i+1]}(s)
=
\underset{x\in[\lambda_{\min},\bar\lambda_{MP}]}{\arg\max}
I(x,s)
{+}\sum_{s_1\in\mathcal S}
\mathbb P(S_{k+1}{=}s_1|S_k{=}s,\bar\lambda^{[i]}(s))v^{[i]}(s_1),
\end{align}
where we have used Theorem \ref{thm1} to restrict the optimization over $[\lambda_{\min},\bar\lambda_{MP}]$, and,
using~(\ref{txprob}),
 \begin{align}
\bar\lambda^{[i+1]}(s)
&{=}
\underset{x\in[\lambda_{\min},\bar\lambda_{MP}]}{\arg\max}
I(x,s)
+\delta\mu(s)v^{[i]}(s-1)
%\nonumber\\&
+\delta A(s)xv^{[i]}(s+1)
+[1{-}\delta\mu(s){-}\delta A(s)x]v^{[i]}(s)
\nn\\&
{=}
\underset{x\in[\lambda_{\min},\bar\lambda_{MP}]}{\arg\max}
I(x,s)+ A(s)xD^{[i]}(s+1).
\label{mavadarviaelcul}
\end{align}
Since the objective function in (\ref{mavadarviaelcul}) is concave in $x$, we have the following cases, yielding~(\ref{PIS}):

\noindent {\bf 1}) $I^\prime(\bar\lambda_{MP},s)+ A(s)D^{[i]}(s+1)\geq 0$,
where $I^\prime(x,s)$ is given by (\ref{derivis}),
 or equivalently, using the  expression of $\bar\lambda_{MP}$ in  (\ref{MPexpr}),
$D^{[i]}(s+1)\geq 0$. In this case,
 the objective function in (\ref{mavadarviaelcul})  is an increasing function of $x$, hence the optimal is $\bar\lambda^{[i+1]}(s)=\bar\lambda_{MP}$.

\noindent {\bf 2}) $I^\prime(\lambda_{\min},s){+}A(s)D^{[i]}(s{+}1)\leq 0$, or equivalently
$D^{[i]}(s{+}1){\leq}\log_2(e){+}\log_2(\rho)\frac{1}{1-\rho}$. In this case,
 the objective function in (\ref{mavadarviaelcul})  is a decreasing  function of $x$, hence the optimal is $\bar\lambda^{[i+1]}(s)=\lambda_{\min}$.

\noindent {\bf 3}) Otherwise, the optimal $\bar\lambda^{[i+1]}(s)$ is the unique $x\in[\lambda_{\min},\bar\lambda_{MP}]$ such that 
$I^\prime(x,s)+ A(s)D^{[i]}(s+1)=0$. Using the expression of $I^\prime(x,s)$ (\ref{derivis})
and of $\bar\lambda_{MP}$ in  (\ref{MPexpr}),
 we obtain
 \begin{align}
  \bar\lambda^{[i+1]}(s)
  =\frac{\lambda_{\max}}{e}\rho^{-\frac{\rho}{1-\rho}}2^{D^{[i]}(s+1)}
  =\bar\lambda_{MP}2^{D^{[i]}(s+1)}.
 \end{align}
The specific form of the policy iteration Algorithm \ref{PIA} is thus proved.

\bibliographystyle{IEEEtran}
\bibliography{IEEEabrv,Refs}

\end{document}